\documentclass[a4paper]{article}
\usepackage{latexsym,amsthm,amsmath,amssymb}

\theoremstyle{plain}
\newtheorem{theorem}{Theorem}
\newtheorem*{theorem*}{Theorem}
\newtheorem{lemma}[theorem]{Lemma}
\newtheorem*{lemma*}{Lemma}
\newtheorem{corollary}[theorem]{Corollary}
\newtheorem{observation}[theorem]{Observation}
\newtheorem*{observation*}{Observation}

\newtheorem*{conjecture*}{Conjecture}
\newtheorem{rrule}{Reduction Rule}
\theoremstyle{definition}
\newtheorem{definition}[theorem]{Definition}

\usepackage{authblk}
\usepackage[utf8x]{inputenc}
\usepackage[english]{babel}
\usepackage[numbers,sort]{natbib}
\usepackage{graphicx}
\usepackage{svg}
\usepackage{pbox}
\usepackage{type1ec}
\usepackage[T1]{fontenc}
\usepackage{bbm}
\usepackage{mathtools}
\usepackage{amssymb}
\usepackage{amsmath}
\usepackage{amsthm}
\usepackage{amsfonts}
\usepackage{pifont}
\usepackage{url}
\usepackage{multirow}
\usepackage{multicol}
\usepackage{fancybox}
\usepackage{colortbl}
\usepackage{soul}
\usepackage{color}
\usepackage{tikz}
\usepackage{tkz-berge}
\usepackage{rotating}
\usepackage{pdflscape}
\usepackage{afterpage}
\usepackage{comment}
\usepackage{etoolbox}
\usepackage{environ}
\usepackage{dirtytalk}

\usepackage{algorithm}
\usepackage{algorithmic}

\usetikzlibrary{decorations,arrows,shapes}
\newcommand{\inners}{1.2pt}
\newcommand{\outers}{1pt}

\newcommand{\td}[1]{\mathbb{#1}}

\usepackage{complexity}
\newclass{\Hard}{hard}
\newclass{\pNP}{paraNP}
\newclass{\Hness}{hardness}
\newcommand{\NPH}{\NP\text{-}\Hard}
\newcommand{\pNPH}{\pNP\text{-}\Hard}
\newcommand{\pNPHness}{\pNP\text{-}\Hness}
\newclass{\Complete}{complete}
\newclass{\Cness}{completeness}
\newcommand{\NPc}{\NP-\Complete}

\newfunc{\YES}{YES}
\newfunc{\NOi}{NO}
\newfunc{\tw}{tw}
\newcommand{\pname}[1]{\textsc{#1}}
\newcommand{\WH}[1]{\W[#1]\text{-}\Hard}
\newcommand{\WHness}[1]{\W[#1]\text{-}\Hness}

\newfunc{\dist}{dist}
\newfunc{\diam}{diam}

\newfunc{\indtree}{indtree}
\newfunc{\cycle}{cycle}
\newfunc{\conn}{conn}
\newfunc{\MSOx}{MSO}
\newcommand{\MSO}[1]{$\MSOx_{#1}$}

\newfunc{\Rep}{Rep}
\newfunc{\DPx}{DP}

\newfunc{\leavesx}{leaves}
\newcommand{\leaves}[1]{\leavesx\left(#1\right)}

\newfunc{\opt}{opt}
\newfunc{\rmcx}{rmc}
\newcommand{\union}{\uplus}
\newcommand{\bigunion}{\biguplus}
\newfunc{\ins}{ins}
\newfunc{\shift}{shift}
\newfunc{\glue}{glue}
\newfunc{\proj}{proj}
\newfunc{\joinf}{join}

\newcommand{\join}{\sqcup}
\newcommand{\meet}{\sqcap}
\newcommand{\rmc}[1]{\rmcx\left(#1\right)}

\newcommand{\bigO}[1]{{\mathcal{O}\!\left(#1\right)}}

\newcommand{\nproblem}[3]{{\centering\fbox{\pbox{\textwidth}{\pname{#1}\\\textit{Instance}: #2\\\textit{Question}: #3}}}}

\NewEnviron{sproof}[1] {
    \begin{proof}[Proof of safeness of Rule~#1]\BODY\end{proof}
}

\title{FPT and kernelization algorithms for the k-in-a-tree problem}
\date{}

\author[1]{Guilherme C. M. Gomes}
\author[1]{Vinicius F. dos Santos}
\author[2]{Murilo V. G. da Silva}
\author[3,4]{Jayme L. Szwarcfiter}

\affil[1]{Departamento de Ci\^encia da Computa\c{c}\~{a}o, Universidade Federal de Minas Gerais -- Belo Horizonte, Brazil}
\affil[2]{Departamento de Inform\'atica, Universidade Federal do Paran\'{a} -- Curitiba, Brazil}
\affil[3]{Universidade Federal do Rio de Janeiro -- Rio de Janeiro, Brazil}
\affil[4]{Universidade do Estado do Rio de Janeiro -- Rio de Janeiro, Brazil}

\begin{document}

\maketitle

\begin{abstract}
    The \pname{three-in-a-tree} problem asks for an induced tree of the input graph containing three mandatory vertices.
    In 2006, Chudnovsky and Seymour [Combinatorica, 2010] presented the first polynomial time algorithm for this problem, which has become a critical subroutine in many algorithms for detecting induced subgraphs, such as beetles, pyramids, thetas, and even and odd-holes.
    In 2007, Derhy and Picouleau [Discrete Applied Mathematics, 2009] considered the natural generalization to $k$ mandatory vertices, proving that, when $k$ is part of the input, the problem is \NPc, and ask what is the complexity of \pname{four-in-a-tree}.
    Motivated by this question and the relevance of the original problem, we study the parameterized complexity of $k$-\pname{in-a-tree}.
    We begin by showing that the problem is \WH{1}\ when jointly parameterized by the size of the solution and minimum clique cover and, under the Exponential Time Hypothesis, does not admit an $n^{o(k)}$ time algorithm.
    Afterwards, we use Courcelle's Theorem to prove fixed-parameter tractability under cliquewidth, which prompts our investigation into which parameterizations admit single exponential algorithms; we show that such algorithms exist for the unrelated parameterizations treewidth, distance to cluster, and distance to co-cluster.
    In terms of kernelization, we present a linear kernel under feedback edge set, and show that no polynomial kernel exists under vertex cover nor distance to clique unless $\NP \subseteq \coNP/\poly$.
    Along with other remarks and previous work, our tractability and kernelization results cover many of the most commonly employed parameters in the graph parameter hierarchy.
\end{abstract}



\section{Introduction}

Given a graph $G = (V,E)$ and a subset $K \subseteq V(G)$ of size three -- here called the set of terminal vertices -- the \pname{three-in-a-tree} problem consists of finding an induced tree of $G$ that connects $K$.
Despite the novelty of this problem, it has become an important tool in many detection algorithms, where it usually accounts for a significant part of the work performed during their executions.
It was first studied by \citeauthor{three_in_a_tree_comb}~\cite{three_in_a_tree_comb} in the context of theta and pyramid detection, the latter of which is a crucial part of perfect graph recognition algorithms~\cite{perfect_graphs} and the former was an open question of interest~\cite{theta_prism_chud}.
Across more than twenty pages, \citeauthor{three_in_a_tree_comb} characterized all pairs $(G,K)$ that do not admit a solution, which resulted in a $\bigO{mn^2}$ time algorithm for the problem on $n$-vertex, $m$-edge graphs.
Since then, \pname{three-in-a-tree} has shown itself as a powerful tool, becoming a crucial subroutine for the fastest known even-hole~\cite{even_hole}, beetle~\cite{even_hole}, and odd-hole~\cite{odd_hole} detection algorithms; to the best of our knowledge, these algorithms often rely on reductions to multiple instances of \pname{three-in-a-tree}, e.g. the theta detection algorithm has to solve $\bigO{mn
^2}$ \pname{three-in-a-tree} instances to produce its output~\cite{super_three_in_a_tree}.
Despite its versatility, \pname{three-in-a-tree} is not a silver bullet, and some authors discuss quite extensively why they think \pname{three-in-a-tree} cannot be used in some cases~\cite{net_subdivision,unichord_free}.
Nevertheless, \citeauthor{super_three_in_a_tree}~\cite{super_three_in_a_tree} very recently made a significant breakthrough and managed to reduce the complexity of \citeauthor{three_in_a_tree_comb}'s algorithm for \pname{three-in-a-tree} to $\bigO{m \log^2 n}$, effectively speeding up many major detection algorithms, among other improvements to the number of \pname{three-in-a-tree} instances required to solve some other detection problems.

As pondered by \citeauthor{super_three_in_a_tree}~\cite{super_three_in_a_tree}, the usage of \pname{three-in-a-tree} as a go-to solution for detection problems may, at times, seem quite unnatural.
In the aforementioned cases, one could try to tackle the problem by looking for constant sized minors or disjoint paths between terminal pairs and then resort to \citeauthor{disjoint_path}'s~\cite{disjoint_path} quadratic algorithm to finalize the detection procedure.
The problem is that neither the minors nor the disjoint paths are guaranteed to be induced; to make the situation truly dire, this constraint makes even the most basic problems \NPH.
For instance, \citeauthor{bienstock}~\cite{bienstock, bienstock_corr} proved that \pname{two-in-a-hole} and \pname{three-in-a-path} are \NPc.
As such, it is quite surprising that \pname{three-in-a-tree} can be solved in polynomial time and be of widespread importance.
It is worth to note that the induced subgraph constraint is also troublesome from the parameterized point of view. \pname{Maximum Matching}, for instance, can be solved in polynomial time~\cite{edmonds}, but if we impose that the matching must be induced subgraph, the problem becomes \WH{1}\ when parameterized by the minimum number of edges in the matching~\cite{induced_matching}.

Naturally, we may wonder how far we may push for polynomial time algorithms when considering larger numbers of terminal vertices, i.e we are interested in the complexity of \pname{$k$-in-a-tree} for $k \geq 4$.
The first authors to examine this problem were \citeauthor{induced_trees_complexity}, who proved in \cite{induced_trees_complexity} that \pname{$k$-in-a-tree} is \NPc\ when the number of terminals is part of the input even on planar bipartite cubic graphs of girth four, but solvable in polynomial time if the girth of the graph is larger than the number terminals.
A few years later, \citeauthor{four_in_a_tree_triangle}
~\cite{four_in_a_tree_triangle} showed that \pname{four-in-a-tree} is solvable in triangle-free graphs, while \citeauthor{k_in_a_tree_girth}
~\cite{k_in_a_tree_girth} proved that so is \pname{$k$-in-a-tree} on graphs of girth at least $k \geq 5$; their combined results imply that \pname{$k$-in-a-tree} on graphs of girth at least $k$ is solvable in polynomial time.
In terms of the \pname{$k$-in-a-path} problem,
\citeauthor{induced_trees_complexity}~\cite{induced_trees_complexity} argued that their hardness reduction also applies to this problem and showed that \pname{three-in-a-path} is \NPc\ even on graphs of maximum degree three.
\citeauthor{k_in_a_path_claw}~\cite{k_in_a_path_claw} proved that \pname{$k$-in-a-path}, \pname{$k$-Induced Disjoint Paths}, and \pname{$k$-in-a-cycle} can be solved in polynomial time on claw-free graphs for every fixed $k$, but all of them are \NPc\ when $k$ is part of the input even on line graphs; in fact, they proved that the previous problems are in \XP\ when parameterized by the number of terminals on claw-free graphs.
Another related problem to \pname{$k$-in-a-tree} is the well known \pname{Steiner Tree}, where we want to find a subtree of the input with cost at most $w$ connecting all terminals.
Being one of Karp's 21 \NPH\ problems~\cite{karp_21}, \pname{Steiner Tree} has received a lot of attention over the decades.
Relevant to our discussion, however, is its parameterized complexity.
When parameterized by the number of terminals, it admits a single exponential time algorithm~\cite{steiner_tree}; the same was proven to be true when treewidth~\cite{treewidth} is the parameter~\cite{lattice_algebra}.
On the other hand, when parameterized by cliquewidth~\cite{cliquewidth}, it is \pNPH\ since it is \NPH\ even on cliques: we may reduce from \pname{Steiner~Tree} itself and add, for each non-edge of the input, an edge of cost $w+1$.
As we see below, our first two results are in complete contrast with the parameterized complexity of \pname{Steiner~Tree}.

\smallskip\noindent
\textbf{Our results.} We concern ourselves with the parameterized complexity of \pname{$k$-in-a-tree}.
We begin by presenting some algorithmic results for \pname{$k$-in-a-tree} in Section~\ref{sec:natural_cw}, showing that the latter is \WH{1}\ when simultaneously parameterized by the number of vertices in the solution and size of a minimum clique cover and, moreover, does not admit an $n^{o(k)}$ time algorithm unless the Exponential Time Hypothesis~\cite{eth} (ETH) fails.
This partially answers a (generalization) of \citeauthor{induced_trees_complexity}'s question about the complexity of \pname{$k$-in-a-tree}, in the sense that there is very little hope of obtaining an algorithm that runs in polynomial time only on the size of the input graph.
On the positive side, we prove tractability under cliquewidth using Courcelle's Theorem~\cite{courcelle_book}, which prompts us, in Section~\ref{sec:single_exp}, to turn our attention to which parameters allow us to devise single exponential time algorithms for \pname{$k$-in-a-tree}.
Using \citeauthor{lattice_algebra}'s dynamic programming optimization machinery~\cite{lattice_algebra}, we show that such algorithms exist under treewidth, distance to cluster, and distance to co-cluster.
In Section~\ref{sec:fes}, we present a kernel with $16q$ vertices and $17q$ edges when we parameterize \pname{$k$-in-a-tree} by the size $q$ of a minimum feedback edge set.
In Section~\ref{sec:klb} we prove that the problem does not admit a polynomial kernel when parameterized by bandwidth, nor when simultaneously parameterized by the size of the solution, diameter, and distance to any graph class of your choosing.
In particular, the latter shows that \pname{$k$-in-a-tree} does not admit a polynomial kernel when parameterized by vertex cover nor when parameterized by distance to clique.
All our negative kernelization results are obtained assuming $\NP \nsubseteq \coNP/\poly$.
In terms of tractability and kernelization, our results encompass most of the commonly employed parameters of \citeauthor{gpp}'s graph parameter hierarchy~\cite{gpp}; we present a summary of our results in Figure~\ref{fig:diagram}.
To see why the distance to solution parameter sits between vertex cover and feedback vertex set, we refer to the end of Section~\ref{sec:natural_cw}.

\begin{figure}[!htb]
    \hspace{-0.52cm}
    \includegraphics[scale=0.5]{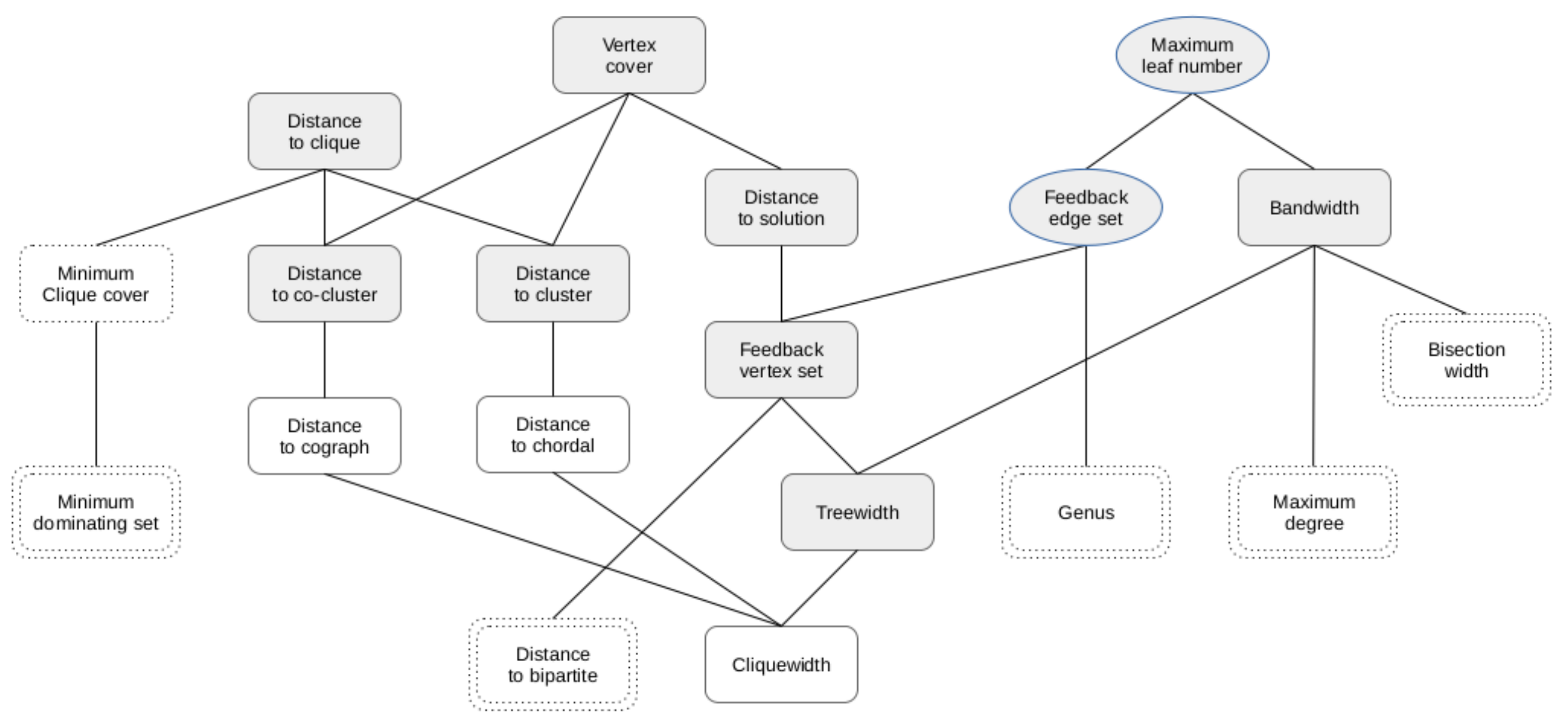}
    \caption{Hasse diagram of graph parameters and associated results for \pname{$k$-in-a-tree}. 
    Parameters surrounded by shaded ellipses have both single exponential time algorithms and polynomial kernels.
    Solid boxes represent parameters under which the problem is \FPT\ but does not admit polynomial kernels; if the box is shaded, we have a single exponential time algorithm for that parameterization. A single dashed box corresponds to a \WH{1} parameterization, while double dashed boxes surround parameters under which the problem is \pNPH.
    Aside from the \pNPHness\ for genus, maximum degree, and distance to bipartite, all results are original contributions proposed in this work.}
    \label{fig:diagram}
\end{figure}

\section{Preliminaries}
\label{sec:prelim}

We refer the reader to~\cite{cygan_parameterized} for basic background on parameterized complexity, and recall here only some basic definitions.
A \emph{parameterized problem} is a language $L \subseteq \Sigma^* \times \mathbb{N}$. 
For an instance $I=(x,q) \in \Sigma^* \times \mathbb{N}$, $q$ is called the \emph{parameter}. 
A parameterized problem is \emph{fixed-parameter tractable} (\FPT) if there exists an algorithm $\mathcal{A}$, a computable function $f$, and a constant $c$ such that given an instance $(x,q)$, $\mathcal{A}$ correctly decides whether $I \in L$ in time bounded by $f(q) \cdot |I|^c$; in this case, $\mathcal{A}$ is called an \emph{\FPT\ algorithm}.
A fundamental concept in parameterized complexity is that of \emph{kernelization}; see~\cite{book_kernels} for a recent book on the topic.
A kernelization	algorithm, or just \emph{kernel}, for a parameterized problem $\Pi$ takes an instance~$(x,q)$ of the problem and, in time polynomial in $|x| + q$, outputs an instance~$(x',q')$ such that $|x'|, q' \leqslant g(q)$ for some function~$g$, and $(x,q) \in \Pi$ if and only if $(x',q') \in \Pi$.
Function~$g$ is called the \emph{size} of the kernel and may be viewed as a measure of the ``compressibility'' of a problem using polynomial-time pre-processing rules.
A kernel is called \emph{polynomial} (resp. \emph{quadratic, linear}) if $g(q)$ is a polynomial (resp. quadratic, linear) function in $q$.
A breakthrough result of \citeauthor{distillation}~\cite{distillation} gave the first framework for proving that some parameterized problems do not admit polynomial kernels, by establishing so-called \emph{composition algorithms}.
Together with a result of \citeauthor{fortnow_santh}~\cite{fortnow_santh}, this allows to exclude polynomial kernels under the assumption that $\NP \nsubseteq \coNP/\poly$, otherwise implying	a collapse of the polynomial hierarchy to its third level~\cite{uniform_non_uniform}.

All graphs in this work are finite and simple.
We use standard graph theory notation and nomenclature for our parameters; for any undefined terminology in graph theory we refer to~\cite{murty}.
We denote the degree of vertex $v$ on graph $G$ by $\deg_G(v)$, and the set of natural numbers  $\{1, 2, \dots, t\}$ by $[t]$.
A graph is a \emph{cluster graph} if each of its connected components is a clique, while a \emph{co-cluster graph} is the complement of a cluster graph.
The \textit{distance to cluster} (\textit{co-cluster}) of a graph $G$, is the size of the smallest set $U \subseteq V(G)$ such that $G \setminus U$ is a cluster (co-cluster) graph.
As defined in~\cite{cai_split}, a set $U \subseteq V(G)$ is an $\mathcal{F}$-\textit{modulator} of $G$ if $G \setminus U$ belongs to the graph class $\mathcal{F}$.
When the context is clear, we omit the qualifier $\mathcal{F}$.
For cluster and co-cluster graphs, one can decide if $G$ admits a modulator of size $q$ in time \FPT\ on $q$~\cite{clusterFPT}.

\section{Fixed-parameter tractability and intractability}
\label{sec:natural_cw}

While it has been known for some time that \pname{$k$-in-a-tree} is \NPc\ even on planar bipartite cubic graphs, it is not known to be even in \XP\ when parameterized by the natural parameter, the number of terminals. 
We take a first step with a negative result about this parameterization, ruling out the existence of an \FPT\ algorithm unless $\FPT = \W[1]$; in fact, we show for stronger parameterization: the maximum size of the induced tree that should contain the set of $k$ terminal vertices $K$ and the size of a minimum clique cover. 

\begin{theorem}
    \pname{$k$-in-a-tree} is $\W[1]$-\Hard\ when simultaneously parameterized by the number of vertices of the induced tree and size of a minimum clique cover.
    Moreover, unless ETH fails, there is no $n^{o(k)}$ time algorithm for \pname{$k$-in-a-tree.}
\end{theorem}

\begin{proof}
    Our reduction is from \pname{Multicolored Independent Set} parameterized by the number of color classes $\ell$.
    Formally, let $H$ be the input to our source problem such that $V(H)$ is partitioned into $\ell$ color classes $\{V_1, \dots, V_\ell\}$ and each $V_i$ induces a clique on $H$.
    Our instance of \pname{$k$-in-a-tree} $(G, K)$ is such that $K = \{v_0, \dots, v_\ell\}$, $V(G) = K \cup V(H)$, every edge of $H$ is also in $G$, each $v_i \in K$ has $N_G(v_i) = V_i$, and $N(v_0) = V(H)$, so $k = \ell + 1$.
    
    If $I$ is a solution to \pname{Multicolored Independent Set}, $I \cup K$ is a solution to $(G, K)$: there are no cycles since $I$ is also an independent set of $G$, $v_0$ connects all vertices of $I$, and each other terminal is connected to exactly one vertex of $I$.
    For the converse, if $T$ is a solution to $(G, K)$, we claim that $I = T \setminus K$ is an independent set of size $\ell$. To see that this is the case, note that: (i) $T \cap V(H)$ must be independent, otherwise they would form a triangle with $v_0$; and (ii) $|T \cap V_i| = 1$ because each $V_i$ is a clique and the only way to connect $v_i$ to $v_0$ is by picking at least one vertex of $V_i$.
    Since $|K| = \ell+1$ we have that the solution size is at most $2k+1$ and, since $Q_i = V_i \cup \{v_i\}$ is a clique, we can cover $G$ with the $k+1$ cliques $\{\{v_0\}\} \bigcup_{i \in [k]} \{Q_i\}$.
    The second statement follows directly from the fact that \pname{Multicolored Independent Set} has no $n^{o(k)}$ time algorithm under ETH~\cite{cygan_parameterized} and that our reduction exhibits a linear relation between the parameters of the source and destination problems.
\end{proof}

Since the natural parameters offer little to no hope of fixed-parameter tractability, to obtain parameterized algorithms we turn our attention to the broad class of structural parameters.
Our first positive result is a direct application of textbook \MSO{1} formulae.

\begin{theorem}
    \pname{$k$-in-a-tree} parameterized by cliquewidth is in \FPT.
\end{theorem}

\begin{proof}
    Let $(G,K)$ be the input to \pname{$k$-in-a-tree}, $R$ be the binary relation represented by the formula $\varphi_R(u,v,Y) = u \in Y \wedge v \in Y \wedge (e(u,v) \vee e(v,u))$, and $TC[R;x,y]$ be the reflexive and transitive closure of $R$.
    As such, $\conn(Y) = \forall x,y(x \in Y \wedge y \in Y \Rightarrow TC[R;x,y])$ is true if and only if $G[Y]$ is connected~\cite{courcelle_book}.
    Similarly, formula $\cycle(X) = \exists x,y,z \in X (x \neq y \wedge y \neq z \wedge x \neq z \wedge e(x,z) \wedge e(y,z) \wedge \exists Y \subset X( z \notin Y \wedge x \in Y \wedge y \in Y \wedge \conn(Y'))$ is true if and only if $G[X]$ has a cycle~\cite{courcelle_book}.
    Putting the previous two formulae together, formula $\indtree(K) = \exists S(K \subseteq S \wedge \conn(S) \wedge \neg \cycle(S))$ is true if and only if there is some superset of $K$ that is connected and acyclic.
    By Courcelle's Theorem~\cite{courcelle_book}, if $G$ has cliquewidth at most $q$, then one can determine in $f(q)n^\bigO{1}$ time if $G$ satisfies $\indtree(K)$ for some set $K$ of terminals.
\end{proof}

Towards showing that the minimum number of vertices we must delete to obtain a solution sits between feedback vertex set and vertex cover in Figure~\ref{fig:diagram}, let $S \subset V(G)$ be such that $G \setminus S$ is a solution.
First, note that $S$ is a feedback vertex set of $G$; for the other inequality, take a vertex cover $C$ of $G$ and note that placing two vertices of $G \setminus C$ with the same neighborhood in $C$ either generates a cycle in the solution or only one of them suffices -- even if we have many terminals, we need to keep only two of them -- so $|S| \leq |C| + 2^{|C|+1}$.
In terms of \pNPHness\ results, we can easily show that \pname{$k$-in-a-tree} is \pNPH\ when parameterized by bisection width\footnote{The width of a bipartition $(A, B)$ of $V(G)$ is the number of edges between the parts. The bisection width is the minimum width of all bipartitions of $V(G)$ such that $|A| \leq |B| \leq |A|+1$.}: to reduce from the problem to itself, we pick any terminal of the input and append to it a path with as many vertices as the original graph to obtain a graph with bisection width one.
Similarly, when parameterizing by the size of a minimum dominating set and again reducing from \pname{$k$-in-a-tree} to itself, we add a new terminal adjacent to any vertex of $K$ and a universal vertex, which can never be part of the solution since it forms a triangle with the new terminal and its neighbor.

\section{Single exponential time algorithms}
\label{sec:single_exp}

All results in this section rely on the \textit{rank based approach} of \citeauthor{lattice_algebra}~\cite{lattice_algebra}, which requires the additional definitions we give below.
Let $U$ be a finite set, $\Pi(U)$ denote the set of all partitions of $U$, and $\sqsubseteq$ be the coarsening relation defined on $\Pi(U)$, i.e given two partitions $p,q \in \Pi(U)$, $p \sqsubseteq q$ if and only if each block of $q$ is contained in some block of $p$.
It is known that $\Pi(U)$ together with $\sqsubseteq$ form a lattice, upon which we can define the usual \textit{join} operator $\join$ and \textit{meet} operator $\meet$~\cite{lattice_algebra}.
The join operation $p \join q$ outputs the unique partition $z$ where two elements are in the same block of $z$ if and only if they are in the same block of $p$ or $q$.
The result of the meet operation $p \meet q$ is the unique partition such that each block is formed by the non-empty intersection between a block of $p$ and a block of $q$.
Given a subset $X \subseteq U$ and $p \in \Pi(U)$, $p_{\downarrow X} \in \Pi(X)$ is the partition obtained by removing all elements of $U \setminus X$ from $p$, while, for $Y \supseteq U$, $p_{\uparrow Y} \in \Pi(Y)$ is the partition obtained by adding to $p$ a singleton block for each element in $Y \setminus U$.
For $X \subseteq U$, we shorthand by $U[X]$ the partition where one block is precisely $\{X\}$ and all other are the singletons of $U \setminus X$; if $X = \{a,b\}$, we use $U[ab]$.

A set of \textit{weighted partitions} of a ground set $U$ is defined as $\mathcal{A} \subseteq \Pi(U) \times \mathbb{N}$.
To speed up dynamic programming algorithms for connectivity problems, the idea is to only store a subset $\mathcal{A}' \subseteq \mathcal{A}$ that preserves the existence of at least one optimal solution.
Formally, for each possible extension $q \in \Pi(U)$  of the current partitions of $\mathcal{A}$ to a valid solution, the optimum of $\mathcal{A}$ relative to $q$ is denoted by $\opt(q, \mathcal{A}) = \min \{w \mid (p, w) \in \mathcal{A}, p \join q = \{U\}\}$.
$\mathcal{A}'$ \textit{represents} $\mathcal{A}$ if $\opt(q, \mathcal{A}') = \opt(q, \mathcal{A})$ for all $q \in \Pi(U)$.
The key result of \citeauthor{lattice_algebra}~\cite{lattice_algebra} is given by Theorem~\ref{thm:reduce}.

\begin{theorem}[3.7 of~\cite{lattice_algebra}]
    \label{thm:reduce}
    There exists an algorithm that, given $\mathcal{A}$ and $U$, computes $\mathcal{A}'$ in time $|\mathcal{A}|2^{(\omega-1)|U|}|U|^\bigO{1}$ and $|\mathcal{A}'| \leq 2^{|U|-1}$, where $\omega$ is the matrix multiplication constant.
\end{theorem}

A function $f : 2^{\Pi(U) \times \mathbb{N}} \times Z \mapsto 2^{\Pi(U) \times \mathbb{N}}$ is said to \textit{preserve representation} if $f(\mathcal{A}', z) = f(\mathcal{A}, z)$ for every $\mathcal{A}, \mathcal{A}' \in \Pi(U) \times \mathbb{N}$ and $z \in Z$; thus, if one can describe a dynamic programming algorithm that uses only transition functions that preserve representation, it is possible to obtain $\mathcal{A}'$.
In the following lemma, let $\rmc{\mathcal{A}} = \{(p,w) \in \mathcal{A} \mid \nexists (p, w') \in \mathcal{A}, w' < w\}$.

\begin{lemma}[Proposition 3.3 and Lemma 3.6 of~\cite{lattice_algebra}]
    \label{lem:functions}
    Let $U$ be a finite set and $\mathcal{A} \subseteq \Pi(U) \times \mathbb{N}$.
    The following functions preserve representation and can be computed in $|\mathcal{A}|\cdot|\mathcal{B}|\cdot|U|^\bigO{1}$ time.
    
    \begin{itemize}
        \item[\textbf{Union.}] For $\mathcal{B} \in \Pi(U) \times \mathbb{N}$, $\mathcal{A} \union \mathcal{B} = \rmc{\mathcal{A} \cup \mathcal{B}}$.
        \item[\textbf{Insert.}] For $X \cap U = \emptyset$, $\ins(X, \mathcal{A}) = \{(p_{\uparrow X \cup U}, w) \mid (p,w) \in \mathcal{A}\}$.
        \item[\textbf{Shift.}] For any integer $w'$, $\shift(w', \mathcal{A}) = \{(p,w + w') \mid (p,w) \in \mathcal{A}\}$.
        \item[\textbf{Glue.}] Let $\hat{U} = U \cup X$, then $\glue(X, \mathcal{A}) = \rmc{\left\{(\hat{U}[X] \join p_{\uparrow \hat{U}}, w) \mid (p,w) \in \mathcal{A}\right\}}$.
        \item[\textbf{Project.}] $\proj(X, \mathcal{A}) = \rmc{\{(p_{\downarrow \overline{X}}, w) \mid (p,w) \in \mathcal{A}, \forall u \in \overline{X} : \exists v \in \overline{X} : p \sqsubseteq U[uv]\}}$, if $X \subseteq U$.
        \item[\textbf{Join.}] If $\hat{U} = U \cup U'$, $\mathcal{A} \subseteq \Pi(U) \times \mathbb{N}$ and $\mathcal{B} \in \Pi(U') \times \mathbb{N}$, then $\joinf(\mathcal{A}, \mathcal{B}) = \rmc{\{(p_{\uparrow\hat{U}} \join q_{\uparrow \hat{U}}, w + w') \mid (p,w) \in \mathcal{A}, (q, w') \in \mathcal{B}\}}$.
    \end{itemize}
\end{lemma}

Even though our problem is unweighted, we found it convenient to solve a weighted version and of \pname{$k$-in-a-tree}.
We state this slightly more general problem below.

\nproblem{Light Connecting Induced Subgraph}{A graph $G$, a set of $k$ terminals $K \subseteq V(G)$, and two integers $\ell, f$.}{Is there a connected induced subgraph of $G$ on $\ell + k$ vertices and at most $f$ edges that contains $K$?}

Note that an instance $(G, K)$ of \pname{$k$-in-a-tree} is positive if and only if there is some integer $\ell$ where the \pname{Light Connecting Induced Subgraph} instance $(G, K, \ell, \ell + k - 1)$ is positive.
Our goal is to use the number of edges in the solution to \pname{Light Connecting Induced Subgraph} as the cost of a partial solution in a dynamic programming algorithm.
This shall be particularly useful for join nodes during our treewidth algorithm, as we may resort to the optimality of the solution to guarantee that the resulting induced subgraph of a $\join$ operation is acyclic.

\subsection{Treewidth}
\label{sec:tw}

A \textit{tree decomposition} of a graph $G$ is a pair $\td{T} = \left(T, \mathcal{B} = \{B_j \mid j \in V(T)\}\right)$, where $T$ is a tree and $\mathcal{B} \subseteq 2^{V(G)}$ is a family where: $\bigcup_{B_j \in \mathcal{B}} B_j = V(G)$;
for every edge $uv \in E(G)$ there is some~$B_j$ such that $\{u,v\} \subseteq B_j$;
for every $i,j,q \in V(T)$, if $q$ is in the path between $i$ and $j$ in $T$, then $B_i \cap B_j \subseteq B_q$.
Each $B_j \in \mathcal{B}$ is called a \emph{bag} of the tree decomposition.
$G$ has treewidth has most $t$ if it admits a tree decomposition such that no bag has more than $t$ vertices.
For further properties of treewidth, we refer to~\citep{treewidth}.
After rooting $T$, $G_x$ denotes the subgraph of $G$ induced by the vertices contained in any bag that belongs to the subtree of $T$ rooted at bag $x$.
An algorithmically useful property of tree decompositions is the existence of a \emph{nice tree decomposition} that does not increase the treewidth of $G$.

\begin{definition}[Nice tree decomposition]
    A tree decomposition $\td{T}$ of $G$ is said to be \emph{nice} if its tree is rooted at, say, the empty bag $r(T)$ and each of its bags is from one of the following four types:
    \begin{enumerate}
        \item \emph{Leaf node}: a leaf $x$ of $\td{T}$ with $B_x = \emptyset$.
        \item \emph{Introduce vertex node}: an inner bag $x$ of $\td{T}$ with one child $y$ such that $B_x \setminus B_y = \{u\}$.
        \item \emph{Forget node}: an inner bag $x$ of $\td{T}$ with one child $y$ such that $B_y \setminus B_x = \{u\}$.
        \item \emph{Join node}: an inner bag $x$ of $\td{T}$ with two children $y,z$ such that $B_x = B_y = B_z$.
    \end{enumerate}
\end{definition}

\begin{theorem}
    There is an algorithm for \pname{Light Connecting Induced Subgraph} that, given a nice tree decomposition of width $t$ of the $n$-vertex input graph $G$ rooted at the forget node for some terminal $r \in K$, runs in time $2^\bigO{t}n^\bigO{1}$.
\end{theorem}

\begin{proof}
    Let $(G, K, \ell, f)$ be an instance of \pname{Light Connecting Induced Subgraph}.
    For each bag $x$, we compute the table $g_x(S, \ell') \subseteq \Pi(S) \times \mathbb{N}$, where $S \subseteq B_x$ contains the vertices of $B_x$ that must be present in a solution and $\ell'$ is the number of vertices we allow in the induced subgraphs of $G_x$; each weighted partition $(p,w) \in g_x(S, \ell')$ corresponds to a choice of an induced subgraph of $G_x$ with $\ell'$ vertices, $w + |E(G[S])|$ edges, and connected components given by the blocks of $p$.
    If $|S| > \ell'$, we define $g_x(S, \ell') = \emptyset$.
    After every operation, we apply the algorithm of Theorem~\ref{thm:reduce}.
    
    \noindent \textbf{Leaf node.} Since $B_x = \emptyset$, the only possible connecting induced subgraph is precisely the empty graph, so we define:
    
    \begin{equation*}
        \centering
        \hfill g_x(\emptyset, \ell') =
        \begin{cases}
            \{(\emptyset, 0)\}, &\text{ if } \ell' = 0 \text{;}\\
            \emptyset, &\text{ otherwise.}\\
        \end{cases}
    \end{equation*}
    
    \noindent \textbf{Introduce node.} Let $y$ be the child bag of $x$ and $B_x \setminus B_y = \{v\}$. We compute $g_x(S, \ell')$ as follows, where $\mathcal{A}_y(S, \ell', v) = \ins(\{v\}, g_y(S \setminus \{v\}, \ell' - 1))$:

    \begin{equation*}
        \centering
        \hfill g_x(S, \ell') =
        \begin{cases}
            \glue\left(N[v] \cap S, \mathcal{A}_y(S, \ell', v)\right), &\text{ if $v \in S$;}\\
            g_y(S, \ell'), &\text{ if $v \notin S \cup K$.}\\
            \emptyset, &\text{ otherwise.}
        \end{cases}
    \end{equation*}
    
    On the third case above, if $v \in K$ but $v \notin S$, we are not including a terminal into the induced subgraph, thus we cannot accept any partition with support $S$ as valid.
    If $v \notin K \cup S$, then no changes are necessary, since the only vertex of $G_x$ not in $G_y$ is not considered for the solution.
    Finally, for the first case, since $v \in S$, we must extend each partition of $g_y(S \setminus \{v\}, \ell' - 1)$ to the ground set $S$ (which we achieve by the insert operation); however, since we are looking for an induced subgraph, we must use every edge between $v$ and the neighbors of $v$ in $S$, merging the connected components containing them.
    Unlike in some connectivity problems such as \pname{Steiner Tree}, we only count the edges \textit{within} $S$ while processing forget nodes; this shall simplify the join operation considerably, as we do not need to worry about repeatedly counting edges inside the current bag.
    
    \noindent \textbf{Forget node.} Let $y$ be the child bag of $x$ and $v$ be the forgotten vertex. The transition is directly computed by:
    
    \begin{equation*}
        g_x(S, \ell') = g_y(S, \ell') \union \shift(|N(v) \cap S|, \proj(\{v\}, g_y(S \cup \{v\}, \ell')))
    \end{equation*}
    
    If $v$ is not used in a partial solution, $g_y(S, \ell')$ already correctly contains all the partial solutions where $v$ is not used; on the other hand, if $v$ was used in some solution, we must eliminate $v$ from the partitions where it appears; however, we only keep the partitions that do not lose a block, otherwise we would have a connected component (represented by $v$) that shall never be connected to the remainder of the subgraph and, thus, cannot be extended to a valid solution.
    
    \noindent \textbf{Join node.} If $y,z$ are the children of bag $x$, we compose its table by the equation:
    
    \begin{equation*}
        g_x(S, \ell') = \bigunion_{\substack{{\ell_1 + \ell_2\ =\ \ell' + |S|}\\{|S|\ \leq\ \ell_1,\ell_2}}} \joinf(g_y(S, \ell_1), g_z(S, \ell_2))
    \end{equation*}
    
    Where the union operator runs over all integer values satisfying the system $\ell_1 + \ell_2 = \ell' + |S|,\ |S| \leq \ell_1, \ell_2$; since we do not know how many vertices were used on the partial solutions of each subtree, we must try every combination to obtain all the partial solutions for the subtree rooted at $x$.
    Since we force the vertices in $S$ to be present in the solutions to the subtree rooted at bag $x$, combining two partial solutions, represented by $(p, w) \in g_y(S, \ell_1)$ and $(q, w') \in g_y(S, \ell_2)$, corresponds to uniting the set of edges of the respective partial solutions $G(p), G(q)$, which results in a merger of connected components.
    Since the edges of $S$ have not been counted towards the weights $w,w'$, $G(p \join q)$ has exactly $w + w' + |E(G[S])|$ edges.
    This is precisely the definition of the $\joinf$ operation.
    
    In order to obtain the answer to the problem, we look at the child $x$ of the forget node for terminal $r$, that is, the child of the root of the tree, and check if $g_x(\{r\}, \ell + |K|) \neq \emptyset$.
    In the affirmative, note that there is only one entry $(\{\{r\}\}, w) \in g(\{r\}, \ell + |K|)$ and that the graph that connects all the terminals using $\ell + |K|$ vertices has exactly $w$ edges, since $E(G[\{r\}]) 0 \emptyset$.
    
    For an introduce bag $x$ with child $y$, the time taken to compute all entries of $g_x$ is of the order of $\ell \sum_{i=0}^{|B_x|}\binom{|B_x|}{i}2^{\omega i}t^\bigO{1} \leq n (1 + 2^{\omega i})^tt^\bigO{1}$; the term $2^{\omega i}$ comes from the time needed to execute the algorithm of Theorem~\ref{thm:reduce} upon a initial set of size $2^i$.
    For join nodes, the intermediate $\join$ operation may yield a set of size $4^i$, so we have that the tables can be computed in time $\ell^3 \sum_{i=0}^{|B_x|}\binom{|B_x|}{i}2^{(\omega+1) i}t^\bigO{1} \leq (1 + 2^{(\omega+1)})^tn^\bigO{1}$.
\end{proof}

\begin{corollary}
    There is an algorithm for \pname{$k$-in-a-tree} that, given a nice tree decomposition of width $t$ of the $n$-vertex input graph $G$ rooted at the forget node for some terminal $r \in K$, runs in time $2^\bigO{t}n^\bigO{1}$.
\end{corollary}
\subsection{Distance to cluster}

We now show that parameterizing by the distance to cluster $q$ also yields an \FPT\ algorithm.
Throughout this section, $G$ is the input graph, $U$ is the cluster modulator, and $\mathcal{C} = \{C_1, \dots, C_r\}$ are the maximal cliques of $G \setminus U$.
We also use the framework developed by \citeauthor{lattice_algebra}~\cite{lattice_algebra} to optimize our dynamic programming algorithm.

\begin{theorem}
    \label{thm:cluster_fpt}
    There is an algorithm for \pname{$k$-in-a-tree} that runs in time $2^\bigO{q}n^\bigO{1}$ on graphs with distance to cluster at most $q$ graphs.
\end{theorem}

\begin{proof}
    Suppose we are given the instance $(G, K)$ and the $q$-vertex cluster modulator $U \subseteq V(G)$.
    We begin by guessing a subset $K \cap U \subseteq S \subseteq U$ of vertices that will be present in a solution for the problem.
    Now, given $S$, we execute the following pre-processing step: for each clique $C_i \in \mathcal{C}$, we discard all but one vertex of each maximal set of true twins; this way, we limit the size of $C_i^* = C_i \setminus K$ to $2^q$.
    
    An entry of our dynamic programming table $f_S(i, c, \ell) \subseteq \Pi(S) \times \mathbb{N}$ is a set of partial solutions of $G_i = G[K \cup S \bigcup_{j = 1}^i C_j]$, each of which uses exactly $c$ vertices of $C_i^*$ and induces a subgraph of $G_i$ on $\ell$ vertices.
    Note that we cannot use more than two vertices of each clique, so we only consider $c \in \{0,1,2\}$.
    In each $(p,w) \in f_S(i, c, \ell)$, each block of $p$ corresponds to the vertices of $S$ that lie in the same connected component of $G_i$, and $w$ is the number of edges used in the respective induced subgraph.
    Our transition is given by the following equation, where $W(S, X) = |N(X) \cap (S \cup K)| + |E(G[X])|$; if either $c + |K \cap C_i| > 2$, $c > \ell$, or there is some vertex in $K \cap C_i$ that has no neighbor in $S$ and $c = 0$, we define $f_S(i, c, \ell) = \emptyset$.
    
    \begin{equation*}
        \centering
        \hfill f_S(i, c, \ell) = \bigunion_{j = 0}^2 
                            \bigunion_{X \in \binom{C_i^*}{c}} \glue_{W(S, X)}(N_S(X), f_S(i-1, j, \ell - c))
    \end{equation*}
    
    The above defines $f_S$ for all $(i, c, \ell) \in [r] \times \{0,1,2\} \times [n]$; we extend $f_S$ to include the base case $f_S(0, 0, |S \cup K|) = \{(p(S, K), E(G[S \cup K])\}$, where $p(S, K) \in \Pi(S)$ is the partition obtained by gluing together the connected components of $G[S]$ which have a common neighbor in $K$; for all other entries, $f_S(r + 1, c, \ell) = \emptyset$.
    Our goal now is to show that there is a solution to our problem using the vertices of $S$ if and only if $(\{S\}, w) \in f_S(r, c, \ell)$, for some pair $\ell \geq |S \cup K|$, $c \geq |K \cap C_1|$, such that $w = \ell - 1$.
    To do so, we first prove that $f(i, c, \ell)$ contains all partitions of $S$ that represent all possible induced subgraphs of $G_i$ on $\ell$ vertices that use $c$ vertices of $C_i^*$, and that use as few edges as possible.
    
    By induction, suppose that this holds for every entry $f_S(i-1, a, b)$.
    If $c + |K \cap C_i| > 2$ or $c > \ell$, either we want to use more than two vertices of $C_i$, which certainly implies that there is a copy of $K_3$ in the solution, or we want to use more than $\ell$ vertices of $C_i^*$, which is equally impossible, so $f_S(i, c, \ell) = \emptyset$.
    If $c = 0$, we have that a weighted partition $(p,w)$ is valid in $G_i$ if and only if it is valid in $G_{i-1}$, since we use no vertices in $V(G_i) \setminus V(G_{i-1}$; thus, $(p,w) \in f_S(i, 0, \ell)$ if and only if $(p,w) \in f_S(i - 1, j, \ell)$ for some $j \in \{0,1,2\}$.
    Otherwise, suppose we want to add a subset of vertices $X \subseteq C_i^*$ to a partial solution $H$, which is represented by $(p, w) \in f_S(i-1, a, \ell - |X|)$.
    In this case, since $N(X) \setminus C_i \subseteq U$ and $U \cap V(H) \subseteq S$, we have that $X$ can only reduce the number of connected components of $H$ if $N(X)$ intersects two distinct blocks of $p$.
    Thus, the connected components of $G[V(H) \cup X]$ are represented, precisely, by $\glue(N_S(X), p)$ and the only new edges used are those between $X$ and $S \cup K$, and the ones internal to $X$; this is precisely the shift accounted by $W(S, X)$.
    The minimality of $w$ for an entry $(p,w) \in f(i, c, \ell)$ is guaranteed by the $\rmcx$ operation imbued in both $\glue$ and $\union$.
    Note that if there is some entry $(p,w) \in f_S(r, c, \ell)$ for some $c$ such that $p = \{S\}$ and $w = \ell - 1$, this means that there is an induced subgraph of graph that has $S$ contained in a connected component, uses $\ell$ vertices and $\ell - 1$ edges, and thus, must be an induced tree of $G$.
    That it connects all vertices of $K$ follows from the fact that $f_S(0,0,|S \cup K|)$ contains $p(S, K)$ and, in every clique $C_i$ that contains a vertex of $K$ with no neighbor in $S$, we force that at least one vertex of $C_i$ must be picked in a solution.
    
    In terms of complexity, after every $\glue$ or $\union$ operation, we apply the algorithm of Theorem~\ref{thm:reduce}.
    For each tuple $(S, i, c, \ell, j)$, we do so up to $|C_i^*|^2 \leq 2^{2q}$ times per tuple, which implies in a time requirement time of the order of $2^{2q} \cdot 2^{q-1}2^{(\omega - 1)q}q^\bigO{1} \leq 2^{(\omega+3)q}q^\bigO{1}$.
    Since we have $2^qn^\bigO{1}$ tuples, our algorithm runs in $2^{(\omega+4)q}n^\bigO{1}$ time.
\end{proof}

\subsection{Distance to co-cluster}

We can use the result on distance to cluster to solve the problem on graphs with distance to co-cluster at most $q$ without much effort, as we see in the following proposition.

\begin{theorem}
    There is an algorithm for \pname{$k$-in-a-tree} that runs in time $2^\bigO{q}n^\bigO{1}$ where $q$ is the distance to co-cluster.
\end{theorem}

\begin{proof}
    Suppose we are given a co-cluster modulator $U$ of the input graph $G$ and let $\mathcal{I} = \{I_1, \dots, I_r\}$ be the family of independent sets of $G - U$.
    Since $G - U$ is a complete multipartite graph, if we pick vertices of three distinct elements of $\mathcal{I}$, we will form a $K_3$ in the induced subgraph.
    Moreover, for each pair $I_i, I_j \in \mathcal{I}$, at most one of them may have more than one vertex in any solution, the induced subgraph would contain a $C_4$.
    This implies that $K$ can intersect $V(G - U)$ in at most three vertices and at most two independent sets.
    If this intersection has size three, for each $K \cap U \subseteq S \subseteq U$, we can easily verify in polynomial time if $G[S \cup K]$ is a tree.
    Otherwise, for each pair $I_i, I_j \in \mathcal{I}$ such that $K \subseteq U \cup I_i \cup I_j$, we guess which one of them will have more than one vertex in the solution, say $I_i$, and which vertex $v \in I_i$ will be in the solution, with the restriction that $K \subseteq U \cup I_j \cup \{v\}$.
    Now, the graph $G' = G[U \cup I_j \cup \{v\}$] has a cluster modulator $U \cup \{v\}$, and we can apply the algorithm of Theorem~\ref{thm:cluster_fpt} on it to decide if there is an induced tree of $G'$ connecting $K$.
    It follows from the observations that there is a valid induced subtree of $G$ if and only if for some choice $I_i, I_j$ and $v \in I_i$ $G'$ has one such induced subgraph.
\end{proof}
\section{A linear kernel for Feedback Edge Set}
\label{sec:fes}

In this section, we prove that \pname{$k$-in-a-tree} admits a linear kernel when parameterized by the size $q$ of a minimum feedback edge set.
Throughout this section, we denote our input graph by $G$, the set of terminals by $K$, and the tree obtained by removing the edges of a minimum size feedback edge set $F$ by $T(F)$.
Note that, if $G$ is connected and $F$ is of minimum size, $G \setminus F$ is a tree; we may safely assume the first, otherwise we either have that $(G, K)$ is a negative instance if $K$ is spread across multiple connected components of $G$, or there must be some edge of $F$ that merges two connected components of $G \setminus F$ and does not create a cycle, contradicting the minimality of $F$. 
The kernelization algorithm we describe works in two steps: it first finds a feedback edge set $F$ that minimizes the number of edges incident to vertices of degree two in $T(F)$, then compresses long induced paths of $G$.
We denote the set of leaves of a tree $H$ by $\leaves{H}$.

\begin{rrule}
    \label{rrule:deg_one}
    If $G$ has a vertex $v$ of degree one, remove $v$ and, if $v \in K$, add the unique neighbor of $v$ in $G$ to $K$.
\end{rrule}

\begin{sproof}{\ref{rrule:deg_one}}
    Safeness follows from the fact that a degree one vertex is in the solution if and only if its unique neighbor also is.
\end{sproof}

\begin{observation}
    \label{obs:leaf_count}
    After exhaustively applying Rule~\ref{rrule:deg_one}, for every minimum feedback edge set $F$ of $G$, $T(F)$ has at most $2q$ leaves. Moreover, $T(F)$ has at most as many vertices of degree at least three as leaves.
\end{observation}

We begin with any minimum feedback edge set $F$ of $G$.
We partition $T(F) \setminus \leaves{T(F)}$ into $(D_2, D_*)$ according to the degree of the vertices of $G$ in $T(F)$: $v \in D_2$ if and only if $\deg_{T(F)}(v) = 2$.
For $u,v,f \in V(G)$, we say that $u$ \textit{$F$-links} $v$ to $f$ if $v=u$ or if $T(F) \setminus \{u\}$ has no $v-f$ path.
We say that vertices $u,f$ are an \textit{$F$-pair} if the set of internal vertices of the unique $u-f$ path $P_F(u,f)$ of $T(F)$ is entirely contained in $D_2$; we denote the set of internal vertices by $P^*_F(u,f)$.

\begin{rrule}
    \label{rrule:inner_edge}
    Let $u,f,w_1,w_2 \in V(G)$ be such that $u,f$ form an $F$-pair, $w_1 \neq w_2 \neq u \neq w_1$, $w_2$ is the unique neighbor of $f$ that $F$-links it to $u$ and $w_1$ $F$-links $w_2$ to $u$. If $fw_1 \in F$, remove edge $fw_1$ from $F$ and add edge $fw_2$ to $F$.
\end{rrule}

\begin{sproof}{\ref{rrule:inner_edge}}
    Let $F' = F \setminus \{fw_1\}$ and note that  $w_2$ $F$-links $f$ to $w_1$; as such, edge $fw_2$ is in the unique cycle of $G \setminus F'$, so $F'' = F' \cup \{fw_2\}$ is a feedback edge set of $G$ of size $q$.
    Furthermore, $w_2$ is the only vertex that has fewer neighbors in $T(F'')$ than in $T(F)$; since $w_2$ had two neighbors in $T(F)$, and $F''$ is a minimum feedback edge set of $G$, $w_2$ is a leaf of $T(F'')$, so it holds that $\leaves{T(F)} \subset \leaves{T(F'')}$.
\end{sproof}

Reduction Rule~\ref{rrule:inner_edge} guarantees that there are no edges in $F$ between vertices of the paths between $F$-pairs, otherwise we could increase the number of leaves of our tree.

\begin{rrule}
    \label{rrule:d2_to_notleaf}
    Let $f,u,v \in V(G)$ be such that $v \notin \leaves{T(F)} \cup P_F(u,f)$, $u,f$ form an $F$-pair, and $|P_F(u,f)| \geq 4$. If there are adjacent vertices $w_1,w_2 \in P^*_F(u,f)$ with $vw_1 \in F$ and $w_2$ $F$-linking $v$ and $w_1$, remove edge $vw_1$ from $F$ and add edge $w_1w_2$ to $F$.
\end{rrule}

\begin{sproof}{\ref{rrule:d2_to_notleaf}}
    Let $F' = F \setminus \{w_1w_2\}$.
    Since $G \setminus F'$ has one more edge than $T(F)$ and $w_2$ $F$-links $v$ and $w$ (see Figure~\ref{fig:d2_to_notleaf}), the unique cycle of $G \setminus F'$ contains edge $w_1w_2$, so $F'' = F' \cup \{vw_1\}$ is a feedback edge set of $G$ of size $q$.
    Since neither $v$ nor $w$ are leaves of $T(F)$ and $w_2 \in D_2$, $\deg_{T(F'')}(w_2) = 1$, so it holds that  $|\leaves{T(F)}| < |\leaves{T(F'')}|$.
\end{sproof}

\begin{figure}[!htb]
    \centering
        \begin{tikzpicture}[xscale=2, yscale=1.2,rotate=-90]
            \GraphInit[unit=3,vstyle=Normal]
            \SetVertexNormal[Shape=circle, FillColor = black, MinSize=3pt]
            \tikzset{VertexStyle/.append style = {inner sep = \inners,outer sep = \outers}}
            \SetVertexLabelOut
            \draw (-0.15, -0.8) rectangle (0.5, 0.8);
            \Vertex[x=0, y=0.45, Math, Lpos=-90]{w_1}
            \Vertex[x=0, y=-0.05, Math, Lpos=-90]{w_2}
            \Vertex[x=0, y=-0.55, NoLabel]{a}
            
            \draw[dashed] (-1.5, -1.2) rectangle (-0.85, 1.2);
            \Vertex[x=-1, y=0.95, Math, Lpos=90]{f}
            \Vertex[x=-1, y=-0.05, Math, Lpos=90]{v}
            \Vertex[x=-1, y=-1.05, Math, Lpos=90]{u}
            
            \Edges(v,u,a,w_2)
            \Edges(w_1,f)
            \Edge[style={double}](w_1)(v)
            \Edge[style={dotted}](w_1)(w_2)
            
            \node at (0.175, -1.5) {$D_2$};
            \node at (-1.165, -1.5) {$D_*$};
        \end{tikzpicture}
        \caption{Example for Reduction Rule~\ref{rrule:d2_to_notleaf}, where the thick edge $vw_1$ is removed from $F$ and the dotted edge $w_1w_2$ added to $F$.\label{fig:d2_to_notleaf}}
\end{figure}
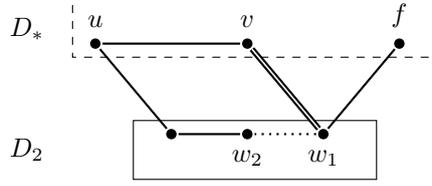

Note that, in Reduction Rule~\ref{rrule:d2_to_notleaf}, the only properties that we exploit are that $v,w_1,w_2 \notin \leaves{T(F)}$ and that there is at least one pair of vertices between $v$ and $f$ in $D_2$.
So even if $v=u$, $u \in D_2$, or $f \in \leaves{T(F)}$, we can apply Rule~\ref{rrule:d2_to_notleaf}.
Essentially, if Rule~\ref{rrule:d2_to_notleaf} is not applicable, for each edge $e \in F$ that contains a vertex $w_1$ of $D_2$ as an endpoint, either $e$ has a leaf of $T(F)$ as its other endpoint, or $w_1$ is adjacent to two vertices not in $D_2$.

\begin{rrule}
    \label{rrule:midway_leaf}
    Let $f,u,v,z,w_2 \in V(G)$ be such that $v \in \leaves{(T(F)} \setminus \{f\}$, $w_2 \in D_2$ is the unique neighbor of $v$ in $T(F)$, $u,f$ and $z,v$ are $F$-pairs, and $u$ $F$-links $f$ to $z$. If there is some $w_1 \in P^*_F(u,f)$ with $vw_1 \in F$, remove $vw_2$ from $F$ and add $vw_1$.
\end{rrule}

\begin{sproof}{\ref{rrule:midway_leaf}}
    Let $F' = F \setminus \{vw_1\}$.
    Since $T(F)$ is a tree and $w_2$ $F$-links $w_1$ and $v$, edge $vw_2$ is contained in the unique cycle of $G \setminus F'$; consequently, $F'' = F' \cup \{vw_2\}$ is a feedback edge set of $G$ of size $q$, but it holds that the degrees of $v$ and $w_2$ in $T(F'')$ are equal to one. 
    Since $w_2 \notin \leaves{T(F)}$, we have that $\leaves{T(F)} \subset \leaves{T(F'')}$.
\end{sproof}

Our analysis for Rule~\ref{rrule:midway_leaf} works even if $u = z$ or $z=w_2$: what is truly crucial is that $w_2 \in D_2$ and that $v \neq f$.
We present an example of the general case in Figure~\ref{fig:midway_leaf}.

\begin{figure}[!htb]
    \centering
        \begin{tikzpicture}[xscale=2, yscale=1.2, rotate=-90]
            \GraphInit[unit=3,vstyle=Normal]
            \SetVertexNormal[Shape=circle, FillColor = black, MinSize=3pt]
            \tikzset{VertexStyle/.append style = {inner sep = \inners,outer sep = \outers}}
            \SetVertexLabelOut
            
            \draw (-0.15, -0.4) rectangle (0.5, 1);
            \Vertex[x=0, y=0.75, Math, Lpos=-90]{w_1}
            \Vertex[x=0, y=-0.15, Math, Lpos=-90, L={w_2}]{w_2}
            
            \draw[thick] (-1.5, 0.8) rectangle (-0.85, 1.9);
            \Vertex[x=-1, y=1.75, Math, Lpos=90]{f}
            \Vertex[x=-1, y=0.95, Math, Lpos=90]{v}
            
            \draw[dashed] (-1.5, -1.2) rectangle (-0.85, -0.1);
            \Vertex[x=-1, y=-0.25, Math, Lpos=90]{u}
            \Vertex[x=-1, y=-1.05, Math, Lpos=90]{z}
            
            \Edges(w_2,z,u,w_1,f)
            \Edge[style={double}](w_1)(v)
            \Edge[style={dotted}](w_2)(v)
            
            \node at (0.175, -0.7) {$D_2$};
            \node at (-1.165, -1.5) {$D_*$};
            \node at (-1.165, 2.5) {$\leaves{T(F)}$};
        \end{tikzpicture}
        \caption{Example for Reduction Rule~\ref{rrule:midway_leaf}, where the thick edge $vw_1$ is removed from $F$ and the dotted edge $vw_2$ is added to $F$.\label{fig:midway_leaf}}
\end{figure}
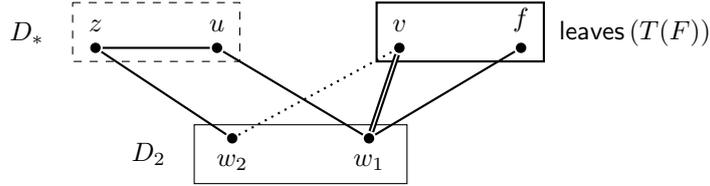

\begin{rrule}
    \label{rrule:ds_to_leaf_to_third}
    Let $f,u,v,z,w_2,w_3 \in V(G)$ be such that $v \in \leaves{(T(F)}$, $w_3 \in N_{T(F)}(u) \cap P_F(u,f)$, $w_2w_3, vz \in E(T(F))$, $u,f$ is an $F$-pair, $z \in D_*$, and $u$ $F$-links $f$ to $v$. If $v$ is adjacent to some $w_1 \in P_F(u,f) \setminus \{w_2\}$ that $F$-links $w_2$ to $f$, remove $vw_1$ from $F$ and add $w_2w_3$ to $F$.
\end{rrule}

\begin{sproof}{\ref{rrule:ds_to_leaf_to_third}}
    Let $F' = F \setminus \{vw_3\}$.
    Since $w_1$ $F$-links $w_2$ to $f$, we have that $w_2$ $F$-links $w_1$ to $v$, so edge $w_2w_3$ belongs to the unique cycle of $G \setminus F'$ and, consequently, $F'' = F' \cup \{w_2w_3\}$ is a feedback edge set of $G$ of size $q$.
    Since $\deg_{T(F)}(w_3) = \deg_{T(F)}(w_2) = 2$, $\deg_{T(F'')}(w_3) = \deg_{T(F'')}(w_2) = 1$, however, $v$ is adjacent to both $z$ and $w_1$ in $T(F'')$, so it holds that $\leaves{T(F'')} = \leaves{T(F)} \cup \{w_2, w_3\} \setminus \{v\}$.
\end{sproof}

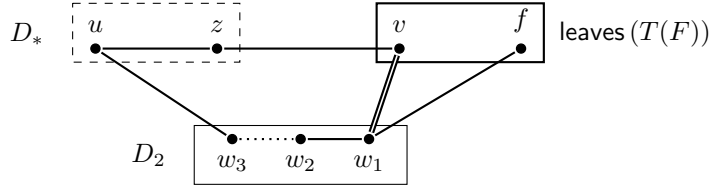
\begin{figure}[!htb]
    \centering
        \begin{tikzpicture}[xscale=2, yscale=1.2, rotate=-90]
            \GraphInit[unit=3,vstyle=Normal]
            \SetVertexNormal[Shape=circle, FillColor = black, MinSize=3pt]
            \tikzset{VertexStyle/.append style = {inner sep = \inners,outer sep = \outers}}
            \SetVertexLabelOut
            
            \draw (-0.15, -0.4) rectangle (0.5, 1);
            \Vertex[x=0, y=0.75, Math, Lpos=-90]{w_1}
            \Vertex[x=0, y=0.3, Math, Lpos=-90]{w_2}
            \Vertex[x=0, y=-0.15, Math, Lpos=-90]{w_3}
            
            \draw[thick] (-1.5, 0.8) rectangle (-0.85, 1.9);
            \Vertex[x=-1, y=1.75, Math, Lpos=90]{f}
            \Vertex[x=-1, y=0.95, Math, Lpos=90]{v}
            
            \draw[dashed] (-1.5, -1.2) rectangle (-0.85, -0.1);
            \Vertex[x=-1, y=-0.25, Math, Lpos=90]{z}
            \Vertex[x=-1, y=-1.05, Math, Lpos=90]{u}
            
            \Edges(v,z,u,w_3)
            \Edges(w_2,w_1,f)
            \Edge[style={double}](w_1)(v)
            \Edge[style={dotted}](w_3)(w_2)
            
            \node at (0.175, -0.7) {$D_2$};
            \node at (-1.165, -1.5) {$D_*$};
            \node at (-1.165, 2.5) {$\leaves{T(F)}$};
        \end{tikzpicture}
        \caption{Example for Reduction Rule~\ref{rrule:ds_to_leaf_to_third}, where the dotted edge $w_2w_3$ is added to $F$ and the thick edge $vw_1$ is removed from $F$.\label{fig:ds_to_leaf_to_third}}
\end{figure}

Our next lemma guarantees that the exhaustive application of rules~\ref{rrule:inner_edge} through~\ref{rrule:ds_to_leaf_to_third} finds a set of paths in $T(F)$ that have many vertices of degree two in $G$; essentially, at this point, we are done minimizing the number of incident edges to vertices of $D_2$.

\begin{lemma}
    Let $a,b \in V(G)$ be an $F$-pair such that $a,b \notin D_2$, $|P^*_F(a,b)| \geq 5$, and let $w$ be one of its inner vertices at distance at least three from both $a$ and $b$.
    If none of the rules between Rule~\ref{rrule:inner_edge} and Rule~\ref{rrule:ds_to_leaf_to_third} are applicable, then $\deg_G(w) = \deg_{T(F)}(w)$.
\end{lemma}

\begin{proof}
    Suppose that this is not the case, and let $v \in N_G(w) \setminus N_{T(F)}(w)$.
    \begin{itemize}
        \item If $v \in P_F(a,b)$, suppose w.l.o.g. that $w$ $F$-links $v$ to $a$; moreover, let $w_2$ be the unique neighbor of $v$ that $F$-links it to $w$.
        In this case, Rule~\ref{rrule:inner_edge} is applicable: $a,v$ are an $F$-pair with the required properties, $w_2$ has the same role here as in the definition of the rule, and we may set $w$ as $w_1$.
        \item If $v \notin \leaves{T(F)}$, we may assume, w.l.o.g., that $w$ $F$-links $v$ to $b$.
        We can apply Rule~\ref{rrule:d2_to_notleaf}: $v$ is not adjacent to $w$ in $T(F)$, so $w$ has one neighbor $w_2 \in D_2$ that $F$-links it to $v$.
        \item If $v \in \leaves{T(F)}$ and its unique neighbor is $w_2 \in D_2 \setminus P_F(a,b)$, we again may assume w.l.o.g. that $w_2$ $F$-links $a$ and $v$. In this case, Rule~\ref{rrule:midway_leaf} is applicable: there is some $z \notin \leaves{T(F)}$ (possibly $z \in \{a, w_2\}$) that forms an $F$-pair with $v$, where $P_F(z,v) \setminus \{z,v\}$ may be empty if $z = w_2$.
        \item If $v \in \leaves{T(F)}$ and $z \in D_*$ is its unique neighbor in $T(F)$, then, since there are at least two other vertices between $w$ and each of the endpoints of $P_F(a,b)$, Rule~\ref{rrule:ds_to_leaf_to_third} is applicable; to see that this is the case, set $w$ to $w_3$ in the definition of the rule and $w_1,w_2$ as appropriate to depending on which endpoint of $P_F(a,b)$ $F$-links $w$ to $v$.
    \end{itemize}
    Thus, we conclude that $v$ cannot exist and that the statement holds.
\end{proof}

At this point, paths between $F$-pairs are mostly the same as in $G$: only the to vertices closest to each endpoint may be adjacent to some leaves of $T(F)$, while all others have degree two in $G$.
We say that $u,f$ are a \textit{strict} $F$-pair if for every $w \in P^*_F(u,f)$, $\deg_G(w) = 2$.

\begin{rrule}
    \label{rrule:glueing}
    Let $u,f \in V(G)$ be a strict $F$-pair. If there are adjacent vertices $w_1, w_2 \in P^*_F(u,f)$ such that either $w_1, w_2 \in K$ or $w_1, w_2 \notin K$, add a new vertex $w^*$ to $G$ that is adjacent to $N_G(w_1) \cup N_G(w_2) \setminus \{w_1, w_2\}$ and remove both $w_1, w_2$ from $G$. If $w_1,w_2 \in K$, set $w^*$ as a terminal vertex.
\end{rrule}

\begin{sproof}{\ref{rrule:glueing}}
    Correctness follows directly from the hypotheses that $w_1 \in K$ if and only if $w_2 \in K$ and that both are degree two vertices. So, in a minimal solution $H$ to $(G,K)$, either both vertices are in $H$ or neither is in $H$.
    For the converse, any minimal solution $H'$ to the reduced instance $(G', K')$ either has $w^*$, in which $H$ is obtained by replacing $w^*$ with both $w_1$ and $w_2$, or $w^* \notin V(H)$, in which case $H'$ itself is a solution to $(G,K)$.
\end{sproof}

\begin{rrule}
    \label{rrule:path_compression}
    Let $u,f \in V(G)$ be a strict $F$-pair such that $P^*_F(u,f) \geq 4$. If Rule~\ref{rrule:glueing} is not applicable, replace $P^*_F(u,f)$ with three vertices $a,t,b$ so that $a$ is adjacent to $u$, $b$ to $f$, and $t$ to both $a$ and $b$. Furthermore, $t$ is a terminal of the new graph if and only if $K \cap P^*_F(u,f) \neq \emptyset$.
\end{rrule}

\begin{sproof}{\ref{rrule:path_compression}}
    Let $G'$ and $K'$ be, respectively, the graph and set of terminals obtained after the application of the rule.
    Suppose $H$ is a minimal solution to the \pname{$k$-in-a-tree} instance $(G,K)$, i.e every vertex of $H$ is contained in a path between two terminals.
    Note that, if $P^*_F(u,f) \cap V(H) = \emptyset$, $H$ is also a solution to the instance $(G', K')$; as such, for the remainder of this paragraph, we may assume w.l.o.g. that $P^*_F(u,f) \cap V(H) \neq \emptyset$ and that $u \in V(H)$.
    If $P_F(u,f) \setminus \{u\} \nsubseteq V(H)$, $H' = H \cup \{a,t\} \setminus P^*_F(u,f)$ is a solution to $(G', K')$: since at least one vertex of $P_F(u,f)$ is not in $V(H)$ and every $w \in P^*_F(u,f)$ has degree two in $G$, the subpaths of $P_F(u,f)$ in $H$ are used solely for the collection of terminal vertices of $P_F(u,f)$; consequently, $H'$ is an induced tree of $G'$ that contains all elements of $K'$.
    On the other hand, if $P_F(u,f) \subseteq V(H)$, $H' = H \cup \{a,t,b\} \setminus P^*_F(u,f)$ is a solution to $(G', K')$; to see that this is the case, note that $H \setminus P^*_F(u,f)$ is a forest with exactly two trees where $u$ and $f$ are in different connected components since $P_F(u,f)$ is the unique path between them in $H$, and $K' \subseteq V(H')$ since $P^*_F \subseteq V(H)$ and $K \setminus P
   ^*_F(u,f) \subseteq V(H) \setminus P^*_F(u,f)$.
    
    For the converse, let $H'$ be a minimal solution to $(G', K')$.
    If $\{a,t,b\} \subseteq V(H')$, $H = H' \cup \{P^*_F(u,f)\} \setminus \{a,t,b\}$ is a solution of $(G,K)$, as we are replacing one path consisting solely of degree two vertices with another that satisfies the same property.
    If $a \in V(H')$ but $b \notin V(H')$, then $t \in K'$ (recall that $H'$ is minimal) and $u \in V(H)$, implying that there is at least one terminal vertex in $P^*_F(u,f)$.
    We branch our analysis in the following subcases, where $v \in P^*_F(u,f) \cap N_G(f)$:
    \begin{itemize}
        \item If $v \notin K$, then $H = H' \cup P^*_F(u,f) \setminus \{v\} \setminus \{a,t\}$ is a solution to $(G,K)$: all terminals of $P^*_F(u,f)$ are contained in $P^*_F(u,v)$ and no cycle is generated since all vertices of the path $P^*_F(u,v)$ have degree two.
        \item If $v \in K$ but $f \notin V(H')$, $H = H' \cup P^*_F(u,f) \setminus \{a,t\}$ is a solution to $(G,K)$: we cannot create any new cycle since $f \notin V(H)$ and $u,f$ form a strict $F$-pair, moreover all terminals of $P^*_F(u,f)$ are contained in $H$.
        \item If $v \in K$ and $f \in V(H')$, there is at least one non-terminal vertex $w \in P^*_F(u,v)$ since Rule~\ref{rrule:glueing} is not applicable to $P_F(u,f)$. As such, we set $H = H' \cup P^*_F(u,w) \cup P^*_F(w,f) \setminus \{a,t\}$ and obtain a solution to $(G,K)$.
    \end{itemize}
    Finally, if $\{a,t,b\} \cap V(H') = \emptyset$, it follows immediately from the assumption that $H'$ is a solution to $(G', K')$ that $H'$ is also a solution to $(G, K)$.
\end{sproof}

We are now ready to state our kernelization theorem.

\begin{theorem}
    When parameterized by the size $q$ of a feedback edge set, \pname{$k$-in-a-tree} admits a kernel with $16q$ vertices and $17q$ edges that can be computed in $\bigO{q^2 + qn}$ time.
\end{theorem}

\begin{proof}
    Let $G$ be our $n$-vertex input graph and $K$ a set of terminals.
    We begin by applying Rule~\ref{rrule:deg_one} until no degree one vertex remains in $G$.
    Then, we take any feedback edge set $F$ of $G$ -- we can obtain one in $\bigO{n}$ time by listing the set of back edges of a depth-first search tree of $G$ -- and construct $\mathcal{P}_F$ in $\bigO{n}$ time.
    
    Let $\mathcal{P}_F$ be the set of paths between all $F$-pairs such that, for each $P_F(u,f) \in \mathcal{P}_F$ it holds that $u,f \notin D_2$.
    By Observation~\ref{obs:leaf_count}, we have at most $2q$ leaves in $T(F)$ and $2q$ vertices in $D_*$, so there are at most $4q$ paths in $\mathcal{P}_F$.
    
    Each iteration of the first part of the algorithm is described below.
    If there is some path $P_F(u,f) \in \mathcal{P}_F$ with $f \in \leaves{T(F)}$ and $P^*_F(u,f) \neq \emptyset$, check if there is an edge in $F$ between $f$ and one of the vertices of $P^*_F(u,f) \cup \{u\}$; if this is the case, apply Rule~\ref{rrule:inner_edge} in $\bigO{n}$ time and move on to the next iteration; by Observation~\ref{obs:leaf_count}, $|\mathcal{P}_F| \leq 4q$, so we can inspect each path and perform this step in $\bigO{q + n}$ time.
    Otherwise, let $vw_1 \in F$ be such that $w_1 \notin \leaves{T(F)}$, and $w_1 \in P_F(u,f) \in \mathcal{P_F}$.
    If $v \notin \leaves{T(F)}$ and $w_1$ is adjacent to some $w_2 \in P^*_F(u,f)$ that $F$-links to $v$, apply Rule~\ref{rrule:d2_to_notleaf}; we can check if these conditions are satisfied in $\bigO{n}$ time, in particular, $F$-linking is a matter of testing if $w_1$ and $v$ are in the same connected component of $T(F) \setminus \{w_2\}$.
    If, however, $v \in \leaves{T(F)} \setminus \{f\}$, $v$ forms an $F$-pair with $z \in V(G)$, $w_2 \in D_2 \cap N_{T(F)}(v)$, and $u$ $F$-links $f$ to $z$, Rule~\ref{rrule:midway_leaf} is applicable in $\bigO{n}$ time.
    Finally, if $v \in \leaves{T(F)} \setminus \{f\}$ is adjacent to some $z \in D_*$ which $F$-links $u$ and $v$, $w_3 \in P_F(u,f) \cap N_{T(F)}(u)$ is adjacent to $w_2$ which in turn is $F$-linked to $f$ by $w_1$, Rule~\ref{rrule:ds_to_leaf_to_third} is applicable is $\bigO{n}$ time.
    
    If none of the conditions stated in the previous paragraph is satisfied, we stop the algorithm: the number of leaves of $T(F)$ cannot be increased by single edge swaps.
    As to the number of iterations, rules~\ref{rrule:inner_edge} through~\ref{rrule:ds_to_leaf_to_third} guarantee that, when applicable, the number of leaves increases by exactly one.
    Since each one of their applications can be performed in $\bigO{q + n}$ time and we have at most $\max_F |\leaves{T(F)}| \leq 2q$ iterations, the above algorithm can be executed in $\bigO{q^2 + qn}$ time.
    
    Let $\mathcal{P}_\alpha$ be the set of paths of $\mathcal{P}_F$ whose endpoints are strict $F$-pairs.
    For each path $P_F(u,f)$ in $\mathcal{P}_F \setminus \mathcal{P_\alpha}$, let $u',f' \in P_F(u,f)$ be the strict $F$-pair that maximizes $|P_F(u',f')|$.
    Note that $|P_F(u,f) \setminus P_F(u',f')| \leq 4$ if there are leaves adjacent to each of the vertices at distance two from the endpoints; we refer to Figure~\ref{fig:substrict} for an example.
    Finally, define $\mathcal{P}_\beta$ as $P_F(u',f')$ for each $P_F(u,f) \in \mathcal{P}_F \setminus \mathcal{P_\alpha}$.
        
    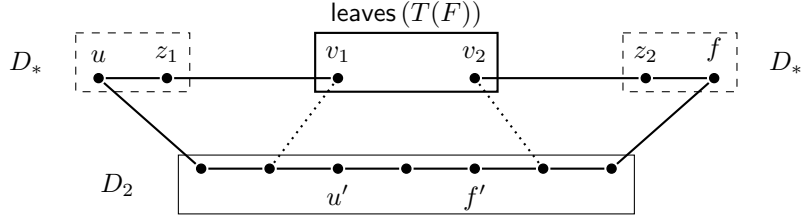
\begin{figure}[!htb]
        \centering
            \begin{tikzpicture}[xscale=2, yscale=1.2, rotate=-90]
                \GraphInit[unit=3,vstyle=Normal]
                \SetVertexNormal[Shape=circle, FillColor = black, MinSize=3pt]
                \tikzset{VertexStyle/.append style = {inner sep = \inners,outer sep = \outers}}
                \SetVertexLabelOut
                
                \draw (-0.15, -1.2) rectangle (0.5, 1.8);
                \node at (0.175, -1.6) {$D_2$};
                \Vertex[x=0, y=1.65, NoLabel]{r2}
                \Vertex[x=0, y=1.20, NoLabel]{r1}
                \Vertex[x=0, y=0.75, Math, Lpos=-90, L={f'}]{fp}
                \Vertex[x=0, y=0.3, Math, NoLabel]{mid}
                \Vertex[x=0, y=-0.15, Math, Lpos=-90, L={u'}]{up}
                \Vertex[x=0, y=-0.60, NoLabel]{l2}
                \Vertex[x=0, y=-1.05, NoLabel]{l1}
                
                \draw[dashed] (-1.5, 1.725) rectangle (-0.85, 2.475);
                \node at (-1.165, 2.8) {$D_*$};
                \Vertex[x=-1, y=2.325, Math, Lpos=90]{f}
                \Vertex[x=-1, y=1.875, Math, Lpos=90]{z_2}

                \draw[dashed] (-1.5, -1.875) rectangle (-0.85, -1.125);
                \node at (-1.165, -2.2) {$D_*$};
                \Vertex[x=-1, y=-1.275, Math, Lpos=90]{z_1}
                \Vertex[x=-1, y=-1.725, Math, Lpos=90]{u}
                
                \draw[thick] (-1.5, -0.3) rectangle (-0.85, 0.9);
                \node at (-1.7, 0.3) {$\leaves{T(F)}$};
                \Vertex[x=-1, y=0.75, Math, Lpos=90]{v_2}
                \Vertex[x=-1, y=-0.15, Math, Lpos=90]{v_1}
                
                \Edges(v_1,z_1,u,l1,l2,up,mid,fp,r1,r2,f,z_2,v_2)
                \Edge[style={dotted}](v_1)(l2)
                \Edge[style={dotted}](r1)(v_2)
            \end{tikzpicture}
            \caption{Example of a path $P_F(u,f)$ and its longest subpath between a strict $F$-pair $u',f'$; dotted edges belong to $F$.\label{fig:substrict}}
    \end{figure}
    
    Now, for each path in $\mathcal{P}_\alpha \cup \mathcal{P_\beta}$, we can apply Rule~\ref{rrule:glueing}; since at each step we remove one vertex from $G$, all applications of the rule amount to $\bigO{n}$-time.
    Afterwards, we apply Rule~\ref{rrule:path_compression} to compress the paths as much as possible; again, this entire process is feasibly done in $\bigO{n}$ steps.
    As such, each path in $\mathcal{P}_\alpha \cup \mathcal{P}_\beta$ has size at most three; consequently each path in $\mathcal{P}_F$ has size at most seven.
    At first glance, this would yield a kernel of size $4q + 7\cdot4q = 32q$; however, we can observe that, for each edge in $F$ incident to a vertex in some path of $\mathcal{P}_\beta$, we are essentially \textit{reducing the number of leaves of $T(F)$ and large degree vertices in one unit each}: the bound of $2q$ leaves is only met with equality if \textit{every} edge of $F$ is incident to two leaves of $T(F)$.
    Therefore, if $\beta = |\mathcal{P}_\beta|$, the kernel's size is given by $|\leaves{T(F)}| + |D_*| + 3|\mathcal{P}_\alpha| + 7|\mathcal{P}_\beta| \leq (2q - \beta) + (2q - \beta) + 3(4q - 2\beta) + 7\beta = 16q - \beta$, which is maximized when $\beta = 0$.
    Regarding the number of edges, the contracted graph has $16q$ vertices and a feedback edge set of size $q$, so it has at most $17q$ edges.
    Finally, since the first part of the algorithm runs in $\bigO{q^2 + qn}$ time and the latter in $\bigO{n}$ time, we have a total complexity of $\bigO{q^2 + qn}$ time.
\end{proof}

\section{Kernelization lower bounds}
\label{sec:klb}

In this section, we apply the cross-composition framework of \citeauthor{cross_composition}~\cite{cross_composition} to show that, unless $\NP \subseteq \coNP/\poly$, \pname{$k$-in-a-tree} does not admit a polynomial kernel under bandwidth, nor when parameterized by the distance to any graph class with at least one member with $t$ vertices for each integer $t$, which we collectively call non-trivial classes.
We say that an \NPH\ problem $R$ \textit{OR-cross-composes} into a parameterized problem $L$ if, given $t$ instances $\{y_1, \dots, y_t\}$ of $R$, we can construct, in time polynomial in $\sum_{i \in [t]}|y_i|$, an instance $(x, k)$ of $L$ that satisfies $k \leq p(\max_{i \in [t]} |y_i| + \log t)$ and admits a solution if and only if at least one instance $y_i$ of $R$ admits a solution; we say that $R$ \textit{AND-cross-composes} into $L$ if the first two conditions hold but all $(x,k)$ has a solution if and only if all $t$ instances of $R$ admit a solution. 

\subsection{Bandwidth}

\begin{theorem}
    When parameterized by bandwidth, \pname{$k$-in-a-tree} does not admit a polynomial kernel unless $\NP \subseteq \coNP/\poly$.
\end{theorem}

\begin{proof}
    We are going to show that \pname{$k$-in-a-tree} AND-cross-composes into itself.
    Let $\mathcal{H} = \{(H_1, K_1), \dots, (H_t, K_t)\}$ be a set of instances of \pname{$k$-in-a-tree} where each graph has $n$ vertices, $\ell \geq 3$ of which are terminals.
    The input $(G, K)$ to \pname{$k$-in-a-tree} parameterized by bandwidth is constructed as follows: $G$ is initially the disjoint union of the $t$ input graphs and $K = \bigcup_{i \in [t]} K_i$; now, for each $i \in [t] $, take two distinct terminals $v_1(i), v_2(i)$ and add edge $v_2(i)v_1(i+1)$ for every $i \in [t-1]$.
    Essentially, we are organizing the $H_i$'s in a path.
    
    Suppose now that every $H_i$ has a solution $T_i$ and note that $T = \bigcup_{i \in [t]} V(T_i)$ is a solution to $(G, K)$: $T$ is a tree and every terminal in $K$ has a path to another.
    For the converse, take a solution $T$ to $(G,K)$ and let $T_i = T \cap V(H_i)$.
    To see that $T_i$ is in fact a solution to $(H_i, K_i)$, it suffices to observe that there can be no path between two vertices of $H_i$ that contains vertices that do not belong to $H_i$.
    As to the bandwidth, we claim that it is at most $n - 1$: we can set each vertex of $H_i$ in the interval $[n(i-1), ni - 1]$ arbitrarily as long as $v_1(i)$ is place at $n(i-1)$ and $v_2(i)$ at $ni-1$, obtaining the mapping  $f$. Consequently, every edge $ab \in E(H_i)$ satisfies $|f(a) - f(b)| \leq n$ and each edge $v_2(i)v_1(i+1)$ satisfies $f(v_1(i+1)) - f(v_2(i)) = n(i + 1 - 1) - (ni - 1) = 1$.
\end{proof}

\subsection{Vertex cover}

In this section, we show that \pname{Hamiltonian Path} on cubic graphs OR-cross-composes into \pname{$k$-in-a-tree} parameterized by vertex cover and number of terminals.
Our construction, however, can be trivially adapted to different parameterizations, such as distance to clique.
In both cases, we heavily rely on the original gadget by Derhy and Picouleau~\cite{induced_trees_complexity}, but  make some modifications to suit our needs.
Let $H$ be an instance of \pname{Hamiltonian Path} on cubic graphs.
The \textit{Derhy-Picouleau graph} of $H$, which we denote by $\DPx(H)$, is constructed as follows:
for each $v_i \in V(H)$, add to $\DPx(H)$ one copy $T_i$ of the gadget depicted in Figure~\ref{fig:vertex_gadget} and, for each edge $v_iv_j \in E(H)$, connect one of the black vertices of $T_i$ to one of the black vertices of $T_j$ so that the degree of each black vertex of $\DPx(H)$ is three.
We say that $T_i$ and $T_j$ are \textit{adjacent} if there is an edge between a black vertex $\alpha_i$ of $T_i$ and a black vertex $\beta_j$ of $T_j$, where $\{\alpha, \beta\} \subset \{a,b,c\}$.
The set of mandatory vertices of $\DPx(H)$ is the set of gray vertices.

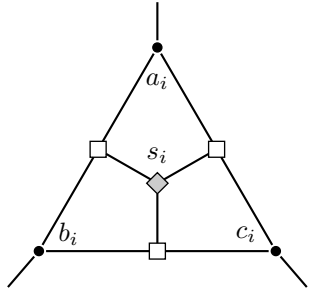
\begin{figure}[!htb]
    \centering
        \begin{tikzpicture}[scale=1.5]
            \GraphInit[unit=3,vstyle=Normal]
            \SetVertexNormal[Shape=circle, FillColor = black, MinSize=3pt]
            \tikzset{VertexStyle/.append style = {inner sep = \inners,outer sep = \outers}}
            \SetVertexLabelOut
            \Vertex[a=90, d=1.2, Lpos=60, NoLabel]{a}
            \Vertex[a=210, d=1.2, Lpos=270, NoLabel]{b}
            \Vertex[a=330, d=1.2, Lpos=270, NoLabel]{c}
            \node at (90:0.9) {$a_i$};
            \node at (210:0.9) {$b_i$};
            \node at (330:0.9) {$c_i$};
            \begin{scope}
                \SetVertexNormal[Shape=coordinate, MinSize=3pt]
                \Vertex[a=90, d=1.6, Lpos=60, NoLabel]{a1}
                \Vertex[a=215, d=1.6, Lpos=270, NoLabel]{b1}
                \Vertex[a=325, d=1.6, Lpos=270, NoLabel]{c1}
                \Edges(a,a1)
                \Edges(b,b1)
                \Edges(c,c1)
            \end{scope}
            \begin{scope}
                \SetVertexNormal[Shape=rectangle, FillColor = white, MinSize=6pt]
                \SetVertexNoLabel
                \Vertex[a=150, d=0.6]{ab}
                \Vertex[a=270, d=0.6]{bc}
                \Vertex[a=30, d=0.6]{ac}
            \end{scope}
            
            \begin{scope}
                \SetVertexNormal[Shape=diamond, FillColor = gray!40, MinSize=7pt]
                \Vertex[x=0, y=0, Lpos=90, Math]{s_i}
            \end{scope}
            
            \Edges(ac,a,ab,b,bc,c,ac,s_i)
            \Edges(ab,s_i,bc)
        \end{tikzpicture}
        \caption{Vertex gadget $T_i$ for vertex $v_i$. Vertex $s_i$ is the only terminal of this gadget; white vertices are part of an independent set of maximum size. \label{fig:vertex_gadget}}
\end{figure}

Before presenting the composition itself, we need to make some slight modifications to $\DPx(H)$, to obtain what we dubbed the \textit{representative graph} of $H$.
Ultimately, our goal is to overlay the multiple instances of \pname{Hamiltonian Path} and, by applying an instance selector gadget, force the graph representing the selected instance to emerge from the confounding structure.

\subsubsection{Representative Graph}

Our key modification to $\DPx(H)$ is to replace the edge between black vertices with edge gadgets.
Suppose that $v_iv_j \in E(H)$, $i < j$, and that $\alpha_i\beta_j \in \DPx(H)$.
We replace the latter edge with the four vertex gadget $e(i, j, \alpha, \beta)$ as in Figure~\ref{fig:edge_gadget}.
Note that $e(i,j, \alpha, \beta)$ and $e(i,j, \beta, \alpha)$ are different gadgets whenever $\alpha \neq \beta$.
By doing this for every edge of $H$, we obtain the \textit{representative graph} of $H$, denoted by $\Rep(H)$. 
Intuitively, if $c_i,b_j$ are in the solution of the  \pname{$k$-in-a-tree} instance given by $\DPx(H)$, then $g_{ij}^{cb}$ is not in the solution of the instance whose input is $\Rep(H)$.
If either $c_i$ or $b_j$ are not in the solution, $g_{ij}^{cb}$ acts as a garbage collector and is used to connect $s_j$ and $s_{ij}^{cb}$.

\begin{figure}[!htb]
    \centering
        \begin{tikzpicture}[scale=1.5]
            \GraphInit[unit=3,vstyle=Normal]
            \SetVertexNormal[Shape=circle, FillColor = black, MinSize=3pt]
            \tikzset{VertexStyle/.append style = {inner sep = \inners,outer sep = \outers}}
            \SetVertexLabelOut
            \begin{scope}[xshift=-2.5cm]
                \Vertex[a=90, d=1.2, Lpos=60, NoLabel]{a1}
                \Vertex[a=210, d=1.2, Lpos=270, NoLabel]{b1}
                \Vertex[a=330, d=1.2, Lpos=270, NoLabel]{c1}
                \node at (90:0.9) {$a_i$};
                \node at (210:0.9) {$b_i$};
                \node at (330:0.9) {$c_i$};
                \begin{scope}
                    \SetVertexNormal[Shape=rectangle, FillColor = white, MinSize=6pt]
                    \SetVertexNoLabel
                    \Vertex[a=150, d=0.6, Lpos=150, Math, L={a_i}]{ab1}
                    \Vertex[a=270, d=0.6, Lpos=270, Math, L={b_i}]{bc1}
                    \Vertex[a=30, d=0.6, Lpos=30, Math, L={c_i}, Ldist=-1pt]{ac1}
                \end{scope}
                
                \begin{scope}
                    \SetVertexNormal[Shape=diamond, FillColor = gray!40, MinSize=7pt]
                    \Vertex[x=0, y=0, Lpos=90, Math,L={s_i}]{s_i1}
                \end{scope}
                
            \begin{scope}
                \SetVertexNormal[Shape=coordinate, MinSize=3pt]
                \Vertex[a=90, d=1.6, Lpos=60, NoLabel]{at1}
                \Vertex[a=215, d=1.6, Lpos=270, NoLabel]{bt1}
                \Edges(a1,at1)
                \Edges(b1,bt1)
            \end{scope}
                
                \Edges(ac1,a1,ab1,b1,bc1,c1,ac1,s_i1)
                \Edges(ab1,s_i1,bc1)
            \end{scope}
            
            \begin{scope}[xshift=2.5cm]
                \Vertex[a=90, d=1.2, Lpos=60, NoLabel]{a2}
                \Vertex[a=210, d=1.2, Lpos=270, NoLabel]{b2}
                \Vertex[a=330, d=1.2, Lpos=270, NoLabel]{c2}
                \node at (90:0.9) {$a_j$};
                \node at (210:0.9) {$b_j$};
                \node at (330:0.9) {$c_j$};
                \begin{scope}
                    \SetVertexNormal[Shape=rectangle, FillColor = white, MinSize=6pt]
                    \SetVertexNoLabel
                    \Vertex[a=150, d=0.6, Lpos=150, Math, L={a_j}]{ab2}
                    \Vertex[a=270, d=0.6, Lpos=270, Math, L={b_j}]{bc2}
                    \Vertex[a=30, d=0.6, Lpos=30, Math, L={c_j}, Ldist=-1pt]{ac2}
                \end{scope}
                
                \begin{scope}
                    \SetVertexNormal[Shape=coordinate, MinSize=3pt]
                    \Vertex[a=90, d=1.6, Lpos=60, NoLabel]{at2}
                    \Vertex[a=325, d=1.6, Lpos=270, NoLabel]{ct2}
                    \Edges(a2,at2)
                    \Edges(c2,ct2)
                \end{scope}
                \begin{scope}
                    \SetVertexNormal[Shape=diamond, FillColor = gray!40, MinSize=7pt]
                    \Vertex[x=0, y=0, Lpos=90, Math,L={s_j}]{s_j}
                \end{scope}
                
                \Edges(ac2,a2,ab2,b2,bc2,c2,ac2,s_j)
                \Edges(ab2,s_j,bc2)
            \end{scope}
            
            \begin{scope}[yshift=-0.2cm]
                \begin{scope}
                    \SetVertexNormal[Shape=diamond, FillColor = gray!40, MinSize=7pt]
                    \Vertex[x=0, y=-1, Math, Lpos=270, L={s_{ij}^{cb}}]{s}
                \end{scope}
                \Vertex[x=-0.6, y=-0.6,NoLabel]{l}
                \Vertex[x=0.6, y=-0.6,NoLabel]{r}
                \Vertex[x=0, y=0.2, Lpos=90, Math, L={g_{ij}^{cb}}]{g}+
                \node at (-0.65, -0.9) {$p_{ij}^{cb}$};
                \node at (0.65, -0.9) {$q_{ij}^{cb}$};
            \end{scope}
            \Edges(c1,l,s,r,b2)
            \Edges(s,g,l)
            \Edges(g,r)
            \Edge[style={bend right=10}](g)(s_j)
        \end{tikzpicture}
        \caption{Edge gadget $e(i,j,c,b)$ for edge $v_iv_j \in E(H)$ ($i < j$).\label{fig:edge_gadget}}
\end{figure}
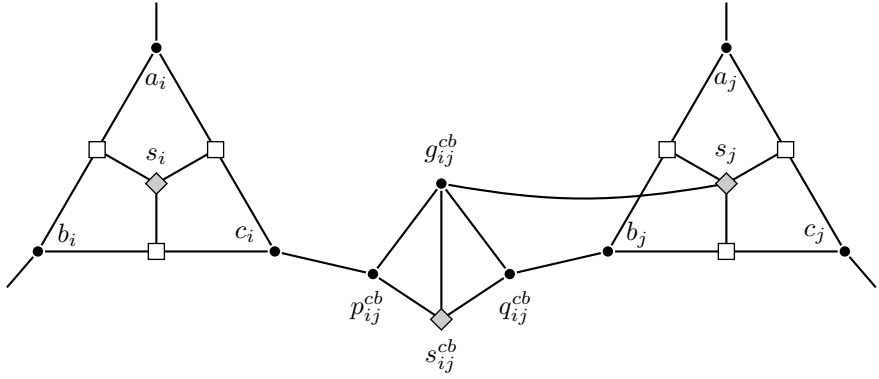

\begin{lemma}
    \label{lem:dp_rep_eq}
    There is an induced tree connecting the terminal vertices of $\DPx(H)$ if and only if there is an induced tree connecting the terminal vertices of $\Rep(H)$.
\end{lemma}

\begin{proof}
    Let $S$ be a solution to $\DPx(H)$; we construct the solution $S'$ to $\Rep(H)$ as follows.
    $S'$ contains every vertex in $S$.
    For every pair of adjacent black vertices $\alpha_i,\beta_j$ add $s_{ij}^{\alpha\beta}$ to $S'$; if both $\alpha_i$ and $\beta_j$ are in $S$, add $\{p_{ij}^{\alpha\beta}, q_{ij}^{\alpha\beta}\}$ to $S'$; otherwise add $g_{ij}^{\alpha\beta}$ to $S'$; this concludes the definition of $S'$.
    To see that $S'$ induces a tree of $\Rep(H)$, note that each path $\langle s_i$, white vertex, $\alpha_i$, $\beta_j$, white vertex, $s_j\rangle$ effectively had edge $\alpha_i\beta_j$ replaced by an induced $P_5$; furthermore, $g_{ij}^{\alpha\beta}$ is in $S'$ if and only if $\{\alpha_i, \beta_j\} \nsubseteq S$, so no cycle can be formed in the edge gadget; since $S$ induces a tree of $\DPx(H)$, we conclude that $S'$ induces a tree of $\Rep(H)$.
    Finally, $S'$ contains all terminal (gray) vertices of $\Rep(H)$: all such vertices also in $\DPx(H)$ were already connected, while the new ones are either included in the induced $P_5$'s with endpoints $\alpha_i,\beta_j$, or are connected by $g_{ij}^{\alpha\beta}$ to $s_j$ (assuming $i < j$).
    
    For the converse, suppose $S' \subset V(\Rep(H))$ induces a tree of $\Rep(H)$.
    Note that we can assume that $S'$ is minimal; in particular, we may safely assume that every black $\alpha_i$ vertex in $S'$ is used to connect $s_i$ to some other $s_j$ or, at the very least, to some terminal of an edge gadget.
    With this restriction in mind, we obtain our solution $S$ to $\DPx(H)$ by setting $S := S' \cap V(\DPx(H))$; note that, if $S'$ is not minimal, there could be a pair of vertices $\alpha_i \in V(T_i)$, $\beta_j \in v(T_j)$ with $\alpha_i\beta_j \in E(\DPx(H))$, which would imply that $S$ was not acyclic.
    Towards showing that $S$ induces a tree of $\DPx(H)$, let $s_i, s_j$ be two terminals of $\DPx(H)$ and $P_\Rep(i,j)$ be the unique path between them in the subgraph of $\Rep(H)$ induced by $S'$.
    We claim that there is no vertex $g_{\ell r}^{\alpha\beta}$ in $P_\Rep(i,j)$:  $g_{\ell r}^{\alpha\beta} \in S'$ implies that it is the unique neighbor of $s_{\ell r}^{\alpha\beta}$ in $S'$, otherwise $S'$ would not induce a tree of $\Rep(H)$.
    Consequently, $P_\Rep(i,j) \cap V(\DPx(H))$ induces a path of $\DPx(H)$ and, since $\bigcup_{i,j \in [n]}P_\Rep(i,j)$ induce a tree of $\Rep(H)$, we conclude that $\bigcup_{i,j \in [n]}P_\Rep(i,j) \cap V(\DPx(H))$ induces a tree of $\DPx(H)$.
\end{proof}

\subsubsection{The OR-cross-composition}

For the remainder of this section, we assume that we are given $t$ instances $\mathcal{H} = \{H_1, \dots, H_t\}$ of \pname{Hamiltonian Path} such that each $H_\ell \in \mathcal{H}$ is cubic and has $n$ vertices, and the input we construct to \pname{$k$-in-a-tree} is the graph $G$.
What we are going to do now is overlay the representative graphs of $\mathcal{H}$ while avoiding additional vertex gadgets and maintaining only a few additional copies of the edge gadgets.
To this end, if $\alpha_i\beta_j \in E(\DPx(H_\ell))$, we say that edge $v_iv_j$ of $H_\ell$ is \textit{represented} by the ordered pair $(\alpha, \beta)$; in a slight abuse of notation, we write $\Rep_\ell(v_iv_j) = (\alpha, \beta)$, where $\{\alpha, \beta\} \subset \{a,b,c\}$.

Formally, $G$ initially has $n$ copies $\{T_1, \dots, T_n\}$ of the vertex gadget given in Figure~\ref{fig:vertex_gadget}.
For each $\ell \in [t]$, define $\mathcal{E}_\ell$ as the set of edge gadgets of $\Rep(H_\ell)$, i.e  $e(i,j,\alpha, \beta) \in \mathcal{E}_\ell$ if and only if $v_iv_j \in E(H)$ and $\Rep_\ell(v_iv_j) = (\alpha, \beta)$.
We update $G$ to include all vertex gadgets and all edge gadgets contained in some $\mathcal{E}_\ell$.
It is critical to note that $\bigcup_{\ell \in [t]} \mathcal{E}_\ell$ has $\bigO{n^2}$ elements.
For our instance selector gadget, we have a copy of $K_{2,t}$ with bipartition $(X, Y)$, where each of the $t$ vertices of $Y$ corresponds to one instance in $\mathcal{H}$, both vertices of $X$ are terminal vertices, and one of them is identified with the terminal vertex $s_1$ of $T_1$.
To conclude the construction of $G$, for each $y_\ell \in Y$ and edge gadget $e(i, j, \alpha, \beta) \notin \mathcal{E}_\ell$, we add all edges between $y_\ell$ and the four vertices of $e(i,j, \alpha, \beta)$, as we show in Figure~\ref{fig:inst_gadget}.

\begin{figure}[!htb]
    \centering
        \begin{tikzpicture}[scale=1.5]
            \GraphInit[unit=3,vstyle=Normal]
            \SetVertexNormal[Shape=circle, FillColor = black, MinSize=3pt]
            \tikzset{VertexStyle/.append style = {inner sep = \inners,outer sep = \outers}}
            \SetVertexLabelOut
            \begin{scope}
                \SetVertexNormal[Shape=diamond, FillColor = gray!40, MinSize=7pt]
                \Vertex[x=0, y=-1, Math, Lpos=270, L={s_{ij}^{cb}}]{s}
            \end{scope}
            \Vertex[x=-0.6, y=-0.6,NoLabel]{l}
            \Vertex[x=0.6, y=-0.6,NoLabel]{r}
            \Vertex[x=0, y=0.2, Lpos=90, Math, L={g_{ij}^{cb}}]{g}
            \node at (-0.65, -0.9) {$p_{ij}^{cb}$};
            \node at (0.65, -0.9) {$q_{ij}^{cb}$};
            \Edges(l,s,r)
            \Edges(s,g,l)
            \Edges(g,r)
        
            \begin{scope}
                \SetVertexNormal[Shape=rectangle, FillColor = white, MinSize=6pt]
                \Vertex[x=2.1,y=0.2, Math, NoLabel]{y_1}
                \Vertex[x=2.1,y=-0.4, Math, Lpos=90, L={y_\ell}]{y_l}
                \Vertex[x=2.1,y=-1, Math, NoLabel]{y_t}
            \end{scope}
            
            \begin{scope}
                \SetVertexNormal[Shape=diamond, FillColor = gray!40, MinSize=7pt]
                \Vertex[x=3.1, y=-0.1, Math, Lpos=0]{s_1}
                \Vertex[x=3.1, y=-0.7, NoLabel]{x}
                \Edges(x,y_1,s_1, y_l, x, y_t, s_1)
            \end{scope}
            \Edges(g, y_l, r)
            \Edge[style={bend right=10}](y_l)(l)
            \Edge[style={bend left=35}](y_l)(s)
        \end{tikzpicture}
        \caption{Interaction between the instance selector gadget and the edge gadget $e(i,j,c,b)$ ($i < j$), if $e(i, j, c, b) \notin \mathcal{E}_\ell$.\label{fig:inst_gadget}}
\end{figure}
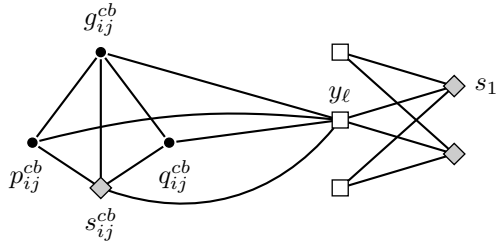

\begin{lemma}
    \label{lem:vc_bound}
    The graph $G$ has a vertex cover of size $\bigO{n^2}$ and at most $\bigO{n^2}$ terminals.
\end{lemma}

\begin{proof}
    Note that $G$ has $7n$ vertices in vertex gadgets, at most $18(n^2 - 3n)$ vertices in edge gadgets, and $t + 1$ vertices in the instance selector gadget (recall that one vertex of $X$ is identified with a terminal of $T_1$).
    Since $Y$ is an independent set, $V(G) \setminus Y$ is a vertex cover with $\bigO{n^2}$ elements.
    For the last part of the statement, it suffices to observe that the set of terminals of $G$ is a subset of $V(G) \setminus Y$.
\end{proof}

\begin{lemma}
    \label{lem:rep_emerge}
    There is no induced tree of $G$ connecting all terminal vertices with zero or more than one vertex of $Y$. Moreover, if $y_\ell \in Y$ is fixed, the graph obtained after removing $X$, $Y$, and all vertices that are in a triangle with $y_\ell$ is precisely $\Rep(H_\ell)$. 
\end{lemma}

\begin{proof}
    If no vertex of $Y$ is picked, then there is no path between the two terminals $s_1, x \in X$ precisely because $N(x) = Y$.
    If two vertices $y_\ell, y_p \in Y$ are picked, then $\{x, y_\ell, s_1, y_p\}$ is an induced $C_4$.
    
    Now, let $e(i,j, \alpha, \beta)$ be an edge gadget of $G$.
    For the second part of the statement, recall that $ V(e(i, j, \alpha, \beta)) \subseteq N(y_\ell) \setminus X$ if and only if $e(i, j, \alpha, \beta) \notin \mathcal{E}_\ell$, i.e the vertices $W_\ell$ of $G$ that form a triangle with $y_\ell$ are precisely those that belong to the edge gadgets $e(i, j, \alpha, \beta)$ \textit{not} present in $\Rep(H_\ell)$.
    As such, the induced subgraph of $G$ that remains after the removal of $W_\ell$, $X$ and $Y$ is composed precisely of the vertex gadgets and $\mathcal{E}_\ell$; since no extra edges were added within either group of gadgets or between them, we have that $G \setminus (W_\ell \cup X \cup Y) = \Rep(H_\ell)$.
\end{proof}

Theorem~\ref{thm:no_kernel_vc} is a direct consequence of Lemmas~\ref{lem:dp_rep_eq},~\ref{lem:vc_bound}, and~\ref{lem:rep_emerge}.

\begin{theorem}
    \label{thm:no_kernel_vc}
    \pname{$k$-in-a-tree} does not admit a polynomial kernel when parameterized by the number of vertices of the induced tree, and size of a minimum vertex cover unless $\NP \subseteq \coNP/\poly$.
\end{theorem}

As for Corollary~\ref{cor:meta_no_kernel}, we observe that there is nothing special about set $Y$ of our instance selector gadget being an independent set; the key feature is that only one of them may be used in a solution; as such, we may freely encode a member of whichever graph classes we are interested in $G[Y]$.

\begin{corollary}
    \label{cor:meta_no_kernel}
    For every non-trivial graph class $\mathcal{G}$, \pname{$k$-in-a-tree} does not admit a polynomial kernel when parameterized by the number of vertices of the induced tree and size of a minimum $\mathcal{G}$-modulator unless $\NP \subseteq \coNP/\poly$.
\end{corollary}

\section{Concluding Remarks}
\label{sec:final}

In this work, we performed an extensive study of the parameterized complexity of \pname{$k$-in-a-tree} and the existence of polynomial kernels for the problem, motivated by the relevance of \pname{three-in-a-tree} in subgraph detection algorithms and a question of \citeauthor{induced_trees_complexity}~\cite{induced_trees_complexity} about the complexity of \pname{four-in-a-tree}.
We presented multiple positive and negative results on the problem, of which we highlight its \WHness{1}\ under the natural parameter, the linear kernel under feedback edge set, and the nonexistence of a polynomial kernel under vertex cover/distance to clique.
The main question about the complexity of \pname{$k$-in-a-tree} for \emph{fixed} $k$, however,  remains open; our hardness result showed that there is no $n^{o(k)}$ time algorithm under ETH, but no \XP\ algorithm is known to exist.
It is worthwhile to revisit some cases where \pname{three-in-a-tree} has not been successful to identify possible applications for \pname{$k$-in-a-tree}.
There are also no known running time lower bounds for the parameters we study, and determining whether or not we can obtain $2^{o(q)}n^\bigO{1}$ time algorithms seems a feasible research direction; still in terms of algorithmic results, it would be quite interesting to see how we can avoid Courcelle's Theorem to get an algorithm when parameterizing by cliquewidth.
The natural investigation of \pname{$k$-in-a-tree} on different graph classes may provide some insights on how to tackle particular cases, such as \pname{four-in-a-tree}; this study has already been started in~\cite{chordal_case} and in others -- such as cographs -- may even be trivial, but many other cases may be quite challenging and much still remains to be done.

\bibliographystyle{plainnat}
\bibliography{refs}

\end{document}